\documentclass[journal,onecolumn]{IEEEtran}

\usepackage{amsmath, amsthm, amssymb}
\usepackage{latexsym}
\usepackage{graphicx}
\usepackage{verbatim}
\usepackage{mathrsfs}
\usepackage{float}
\usepackage[lofdepth, lotdepth]{subfig}
\usepackage{array}
\usepackage{url}
\usepackage{algorithm}
\usepackage[noend]{algorithmic}
\usepackage{cite}
\usepackage[table]{xcolor}
\usepackage{enumerate}
\usepackage{eucal}
\usepackage{stmaryrd}
\usepackage{mathtools}
\usepackage{pgf}
\usepackage{tikz}
\usetikzlibrary{arrows,automata}
\usepackage[latin1]{inputenc}
\usepackage{tikz}
\usetikzlibrary{arrows, arrows.meta,calc,positioning,shapes.arrows}
\usetikzlibrary{decorations.markings}
\usetikzlibrary{chains,fit,shapes}
\usetikzlibrary{mindmap,trees}
\usepackage{ dsfont }

\pgfarrowsdeclarecombine{|<}{>|}{|}{|}{latex}{latex}
\def\Dimline[#1][#2][#3]{
	\begin{scope}[>=latex] 
		\draw[|<->|,
		decoration={markings, 
			mark=at position .5 with {\node at (0,0.25) {\small{#3}};},
		},
		postaction=decorate] #1 -- #2 ;
	\end{scope}
}

\usepackage[top=0.6in, bottom=0.6in, left=0.55in, right=0.55in]{geometry}
\usepackage{mathtools}
\DeclarePairedDelimiter\ceil{\lceil}{\rceil}
\DeclarePairedDelimiter\floor{\lfloor}{\rfloor}

\theoremstyle{plain}
\newtheorem{theorem}{Theorem}

\newtheorem{lemma}{Lemma}

\newtheorem{corollary}{Corollary}

\theoremstyle{definition}
\newtheorem{definition}{Definition}
\newtheorem{example}{Example}
\newtheorem{remark}{Remark}

\DeclareMathAlphabet{\mathbfsl}{OT1}{ppl}{b}{it} 

\newcommand{\bA}{{\mathbfsl A}}
\newcommand{\bB}{{\mathbfsl B}}

\newcommand{\bp}{{\mathbfsl p}}

\newcommand{\by}{{\mathbfsl y}}

\newcommand{\bc}{{\mathbfsl c}}

\newcommand{\bw}{{\mathbfsl{w}}}

\newcommand{\br}{{\mathbfsl{r}}}
\newcommand{\bx}{{\mathbfsl{x}}}
\newcommand{\bz}{{\mathbfsl{z}}}

\newcommand{\ppmod}[1]{~({\rm mod~}#1)}

\renewcommand{\ge}{\geqslant}
\renewcommand{\le}{\leqslant}

\newcommand{\et}{{\emph{et al.}}}

\newcommand{\enc}{\textsc{Enc}}
\newcommand{\dec}{\textsc{Dec}}

\IEEEoverridecommandlockouts
\begin{document}

\pagestyle{empty}

\title{Two-Dimensional RC/Subarray Constrained Codes: Bounded Weight and Almost Balanced Weight \\[2mm]}

\author{
   \IEEEauthorblockN{
   	Tuan Thanh Nguyen,
	Kui Cai, 
	Han Mao Kiah,
	Kees A. Schouhamer Immink,
	and Yeow Meng Chee}
\thanks{ This work was presented in part at the 2021 IEEE International Symposium on Information Theory \cite{ISIT1} and 2022 IEEE International Symposium on Information Theory \cite{ISIT2}. The work of Kui Cai and Tuan Thanh Nguyen is supported by Singapore Ministry of Education Academic Research Fund Tier 2 MOE2019-T2-2-123. The research of Han Mao Kiah is supported by the Singapore Ministry of Education under Grant MOE2019-T2-2-171.}
\thanks{Tuan Thanh Nguyen and Kui Cai are with the Science, Mathematics and Technology Cluster, Singapore University of Technology and Design, Singapore 487372 (email: \{tuanthanh\_nguyen, cai\_kui\}@sutd.edu.sg).}
\thanks{Han Mao Kiah is with the School of Physical and Mathematical Sciences, Nanyang Technological University, Singapore 637371 (email: hmkiah@ntu.edu.sg).}
\thanks{Kees A. Schouhamer Immink is with the Turing Machines Inc, Willemskade 15d, 3016 DK Rotterdam, The Netherlands (email: immink@turing-machines.com).}
\thanks{Yeow Meng Chee is with the Department of Industrial Systems Engineering and Management, National University of Singapore, Singapore (email: ymchee@nus.edu.sg).}
}

\maketitle

\hspace{-3mm}\begin{abstract}
In this work, we study two types of constraints on two-dimensional binary arrays. In particular, given $p,\epsilon>0$, we study
\begin{itemize}
\item The $p$-bounded constraint: a binary vector of size $m$ is said to be $p$-bounded if its weight is at most $pm$, 
\item The $\epsilon$-balanced constraint:  a binary vector of size $m$ is said to be $\epsilon$-balanced if its weight is within $\big[(1/2-\epsilon)m, (1/2+\epsilon)m\big]$.
\end{itemize}

Such constraints are crucial in several data storage systems, those regard the information data as two-dimensional (2D) instead of one-dimensional (1D), such as the crossbar resistive memory arrays and the holographic data storage. In this work, efficient encoding/decoding algorithms are presented for binary arrays so that the weight constraint (either $p$-bounded constraint or $\epsilon$-balanced constraint) is enforced over every row and every column, regarded as 2D row-column (RC) constrained codes; or over every subarray, regarded as 2D subarray constrained codes. 

While low-complexity designs have been proposed in the literature, mostly focusing on 2D RC constrained codes where $p=1/2$ and $\epsilon=0$, this work provides efficient coding methods that work for both 2D RC constrained codes and 2D subarray constrained codes, and more importantly, the methods are applicable for arbitrary values of $p$ and $\epsilon$. Furthermore, for certain values of $p$ and $\epsilon$, we show that, for sufficiently large array size, there exists linear-time encoding/decoding algorithm that incurs at most one redundant bit. 




\end{abstract}


\section{Introduction}

Two-dimensional (2D) weight-constrained codes have attracted recent attention due to various application in modern storage devices that are attempting to increase the storage density by regarding the information data as two dimensional binary arrays. In this work, we are motivated by the application of 2D weight-constrained codes in the {\em holographic recording systems} and the {\em resistive memory} based on crossbar arrays. In particular, in optical recording, the holographic memory capitalizes on the fact that the recording device is a surface and therefore the recording data should be regarded as 2D, as opposed to the track-oriented one-dimensional (1D) recording paradigm \cite{roth:1999, D:1990, J:1996, D:1992}. On the other hand, the key in resistive memory technologies is that the memory cell is a passive two-terminal device that can be both read and written over a simple crossbar structure \cite{Chen:2011, T:2009, S:2014, R:2016}. Both models offer a huge density advantage, however, face new reliability issues and introduce new types of constraints, which are now 2D-constraints, rather than 1D-constraints.
We next briefly describe the motivation of our study on the two types of 2D weight-constraints: the {\em bounded weight-constraint} and the {\em almost-balanced weight-constraint}. 


The 2D bounded weight-constraint is used to limit the number of 1's in an array. It has been suggested as an effective method to reduce the {\em sneak path} effect, a fundamental and challenging problem, in the crossbar memory arrays. The sneak path problem was addressed by numerous works with different approaches and at various system layers \cite{ROTH:ISIT, O:2009, eitan:2016,zhong:2020,G:2020}. In particular, when a cell in a crossbar array is read, a voltage is applied upon it, and current measurement determines whether it is in a low-resistance state (LRS, corresponding to a `1') or a high-resistance state (HRS, corresponding to a `0'). The sneak path occurs when a resistor in the HRS is being read, current also passes through a series of resistors in the LRS exists in parallel to it, thereby causing it to be erroneously read as low-resistance. Therefore, by enforcing fewer memory cells with the LRSs, we can reduce the sneak path effect. This can be achieved by applying constrained coding techniques to convert the user data into 2D-constrained array that limits the number of 1's. Motivated by the application, in this work, we study the 2D $p$-bounded weight-constraint defined over $n\times n$ arrays via two different models: the 2D $p$-bounded RC constrained codes, that limit the number of 1's in every row and every column; and the 2D $p$-bounded $n'$-Subarray constrained codes, that limit the number of 1's in every subarray of size $n' \times n'$ for some $n'<n$. Here, a binary vector of size $m$ is said to be $p$-bounded if its weight is at most $pm$. For 2D $p$-bounded RC constrained codes, Ordentlich and Roth \cite{ROTH:ISIT} required the weight in every row and every column to be at most half, i.e. $p=1/2$, and presented efficient encoders with redundancy at most $2n$ for $n \times n$ arrays. In \cite{ROTH:2000}, the authors studied the bounds of codes that required the weight in every row and every column is precisely $pn$ and provided a coding scheme based on the {\em enumeration coding technique}. In this work, we extend the study of $p$-bounded weight-constraint in the literature works and design 2D $p$-bounded RC constrained codes and 2D $p$-bounded subarray constrained codes  for arbitrary $p\in(0,1)$. 


The 2D almost-balanced weight-constraint is used to control the imbalance between `1's and `0's in an array. This constraint is crucial in holographic recording systems. In such systems, data is stored optically in the form of 2D pages \cite{J:1996, D:1992}, and each data page is a pattern of `0's and `1's, represented by {\em dark} and {\em light} spots, respectively. To improve the reliability of the holographic recording system, one suggested solution by Vardy \et{} \cite{vardy:1996} is to use coding techniques that do not permit a large imbalance between `0's and `1's so that, during recording, the amount of signal light be independent of the data content. In \cite{ROTH:2000}, Ordentlich and Roth also emphasized the importance of balanced codes, and experiments' reports on holographic memory and other existing optical devices also suggested that `0's and `1's in the recorded data need to be balanced within certain area or patterns. Motivated by the application, in this work, we study the 2D {\em $\epsilon$-balanced} weight-constraint defined over $n\times n$ arrays via two different models: the 2D $\epsilon$-balanced RC constrained codes, that enforce every row and every column to be $\epsilon$-balanced; and the 2D $\epsilon$-balanced $m$-subarray constrained codes, that enforce every subarray of size $m \times m$ to be $\epsilon$-balanced for given $m<n$. For 2D $\epsilon$-balanced RC constrained codes, when $\epsilon=0$, Talyansky \et{} \cite{roth:1999} enforced the weight in every row and every column of $n \times n$ array to be exactly $n/2$ and presented an efficient encoding method, that uses roughly $2n\log n+\Theta(n\log \log n)$ redundant bits. To further reduce the redundancy, instead of using one of the algorithms in \cite{knuth:1986,bose:1996}, Talyansky \et{} balanced the rows with the (more computationally complex) enumerative coding technique. Consequently, the redundancy can be reduced to $(3/2)n\log n+\Theta(n\log \log n)$ redundant bits. 
On the other hand, there is no known design for arbitrary value of $\epsilon$. The efficient design of 2D $\epsilon$-balanced RC constrained codes and 2D $\epsilon$-balanced subarray constrained codes, given arbitrary value of $\epsilon>0$, is the main contribution of this work.

In this work, we present two efficient coding methods for 2D RC constrained codes and 2D subarray constrained codes for arbitrary $p\in(0,1)$ and $\epsilon\in(0,1/2)$. The coding methods are based on: (method A) the divide and conquer algorithm and a modification of the Knuth's balancing technique, and (method B) the sequence replacement technique. The coding rate of our proposed methods approaches the channel capacity for all $p,\epsilon$. In addition, for certain values of $p$ and $\epsilon$, we show that for sufficiently large $n$, method B incurs at most one redundant bit.


We first go through certain notations and review prior-art coding techniques.

\section{Preliminaries} 

\subsection{Notations} 

Given two binary sequences $\bx=x_1\ldots x_m$ and $\by=y_1\ldots y_n$, the {\em concatenation} of the two sequences is defined by $\bx \by \triangleq x_1\ldots x_m y_1\ldots y_n.$ For a binary sequence $\bx$, we use ${\rm wt}(\bx)$ to denote the weight of $\bx$, i.e the number of ones in $\bx$. We use $\overline{\bx}$ to denote the complement of $\bx$. For example, if $\bx=00111$ then ${\rm wt}(\bx)=3$ and $\overline{\bx}=11000$. 

Let ${\bA}_n$ denote the set of all $n \times n$ binary arrays. The $i$th row of an array $A\in {\bA}_n$ is denoted by $A_i$ and the $j$th column is denoted by $A^j$. Note that an $n \times n$ binary array $A$ can be viewed as a binary sequence of length $n^2$. We define $\Phi(A)$ as a binary sequence of length $n^2$ where bits of the array $A$ are read row by row. For example, if 
 \[A= \left( \begin{array}{ccc}
a & b & c\\
d & e & f\\
g & h & i\end{array} \right),\text{ then } \Phi(A)=abcdefghi.\]

\begin{definition}
Given $m,\epsilon$ where $\epsilon\in [0,1/2]$, a binary sequence $\bx$ of length $m$ is said to be $\epsilon$-balanced if ${\rm wt}(\bx) \in \big[(1/2-\epsilon)m, (1/2+\epsilon)m\big]$. When $\epsilon=0$, we say the sequence is balanced.
\end{definition}

\begin{definition}
Given $m,p$, where $p\in[0,1]$, a binary sequence $\bx$ of length $m$ is said to be $p$-bounded if ${\rm wt}(\bx)\le pm$. 
\end{definition}


Given $n,m,p,\epsilon$, where $\epsilon\in [0,1/2], p\in[0,1], m<n$, we set 
\begin{small}
\begin{align*}
{\rm B}_{\rm RC}(n;p) &\triangleq \Big\{A \in {\bA}_n:  A_i, A^i \text{ are $p$-bounded for all } 1\le i\le n \Big\},\\
{\rm Bal}_{\rm RC}(n;\epsilon) &\triangleq \Big\{A \in {\bA}_n:  A_i, A^i \text{ are $\epsilon$-balanced for all } 1\le i\le n \Big\}. \\
{\rm B}_{\rm S}(n,m;p) &\triangleq \Big\{A \in {\bA}_n: \text{every subarray } B \text{ of size } m\times m 
\text{ are $p$-bounded} \Big\}, \text{ and } \\
{\rm Bal}_{\rm S}(n,m;\epsilon)  &\triangleq \Big\{A \in {\bA}_n: \text{every subarray } B \text{ of size } m\times m  
\text{ are $\epsilon$-balanced} \Big\}.
\end{align*} 
\end{small}
In this work, we are interested in the problem of designing efficient coding methods that encode (decode) binary data  to (from) ${\rm B}_{\rm RC}(n;p)$, ${\rm Bal}_{\rm RC}(n;\epsilon)$, ${\rm B}_{\rm S}(n,m;p)$ and ${\rm Bal}_{\rm S}(n,m;\epsilon)$.
\begin{definition}
The map $\enc: \{0,1\}^k\to \{0,1\}^{n\times n}$
is a {\em 2D $p$-bounded RC encoder} if $\enc(\bx) \in {\rm B}_{\rm RC}(n;p)$ for all $\bx\in \{0,1\}^k$ and there exists a {\em decoder} map $\dec:\{0,1\}^{n\times n} \to \{0,1\}^k$ such that $\dec\circ\enc(\bx)=\bx$. The {\em coding rate of the encoder} is measured by $k/n^2$ and the {\em redundancy of the encoder} is measured by the value $n^2 -  k$ (bits). The {\em 2D $\epsilon$-balanced RC encoder}, {\em 2D $p$-bounded subarray encoder}, and {\em 2D $\epsilon$-balanced subarray encoder} are similarly defined.
\end{definition}

Our design objectives include low redundancy (equivalently, high code rate) and low complexity of the encoding/decoding algorithms. In this work, we show that for sufficiently large $n$, there exist efficient encoders/decoders for ${\rm B}_{\rm RC}(n;p)$, ${\rm Bal}_{\rm RC}(n;\epsilon)$, ${\rm B}_{\rm S}(n,m;p)$ and ${\rm Bal}_{\rm S}(n,m;\epsilon)$, that incur at most one redundant bit.

\subsection{Literature works on 2D $p$-bounded RC constrained codes ${\rm B}_{\rm RC}(n;p)$}
We briefly describe the literature works by Ordentlich and Roth in \cite{ROTH:ISIT,ROTH:2000} that provided encoding/decoding algorithms for ${\rm B}_{\rm RC}(n;p)$. 
\begin{itemize} 
\item For $p=1/2$, i.e. the weight of every row and every column is at most $n/2$, Ordentlich and Roth \cite{ROTH:ISIT} presented two low complexity coding methods. The first method is based on flipping rows and columns of an arbitrary binary array (i.e. using the complement of rows and columns) until the weight-constraint is satisfied in all rows and columns while the second method is based on efficient construction of {\em antipodal matching}. Both codes have roughly $2n$ redundant bits. A lower bound on the optimal redundancy was shown to be $\lambda n+o(n)$ for a constant $\lambda\approx 1.42515$ in \cite{bound:2011}. Note that these two methods can be used to construct ${\rm B}_{\rm RC}(n;p)$ for arbitrary $p>1/2$, in which ${\rm B}_{\rm RC}(n;1/2) \subset {\rm B}_{\rm RC}(n;p)$. 


\item For $p<1/2$, one may follow the coding method in \cite{ROTH:2000}, based on enumeration coding, that ensure the weight in every row and every column to be precisely $pn$. The redundancy of the proposed encoder was at most $2n \mu(n,p)$, where $\mu(n,p)$ is the least redundancy required to encode one-dimensional binary codewords of length $n$ such that the weight is $pn$. If we set
\begin{equation*}
{\rm Q}(n,p)=\Big\{ \bx \in \{0,1\}^n: {\rm wt}(\bx)=pn \Big\},
\end{equation*} 
then we have $\mu(n,p)=nH(p)-\log|{\rm Q}(n,p)|$ where $H(p)=-p\log p-(1-p)\log (1-p)$. It is easy to verify that $\mu(n,p)=\Theta(n)$ for all $p<1/2$. In addition, Ordentlich and Roth \cite{ROTH:2000} also provided a lower bound on the optimal redundancy, which is at least $2n \mu(n,p)+O(n+\log n)$ bits.  
\end{itemize} 

In this work, we first propose efficient coding methods for ${\rm B}_{\rm RC}(n;p)$ when $p>1/2$ or $p<1/2$. Particularly, when $p<1/2$, the redundancy can be reduced to be at most $n \mu(n,p)+O(n+\log n)$ bits. When $p>1/2$, the redundancy can be reduced significantly to be at most $\Theta(n)$ bits and then to only one bit. We then extend the results to design efficient encoders for 2D $p$-bounded subarray constrained codes ${\rm B}_{\rm RC}(n,m;p)$ when $m=n-o(n)$ and $p\geq 1/2$. We review below the antipodal matching (defined in \cite{ROTH:ISIT}) as it will be used in one of our proposed coding methods. 

\begin{definition}[Ordentlich and Roth \cite{ROTH:ISIT}]\label{def} An antipodal matching $\phi$ is a mapping from $\{0,1\}^n$ to itself with the following properties holding for every $\bx \in \{0,1\}^n$:
\begin{enumerate}
\item ${\rm wt}(\phi(\bx)) = n - {\rm w}(\bx).$
\item If ${\rm wt}(\bx)>n/2$ then $\phi(\bx)$ has all its 1's in positions where $\bx$ has 1's. In other words, suppose $\bx=x_1x_2\ldots x_n$ and $\by=\phi(\bx)=y_1y_2\ldots y_n$, then $y_i =1$ implies $x_i =1$ for  $1\le i\le n$.
\item $\phi(\phi(\bx)) = \bx$.
\end{enumerate}
\end{definition}

Ordentlich and Roth \cite{ROTH:ISIT} presented an efficient construction of antipodal matchings $\phi$ for all $n$. In fact, such an antipodal matching can be decomposed into a collection of bijective mappings $\phi = \cup_{i=0}^n \phi_i$, where 
\begin{small}
\begin{equation*}
\phi_i: \Big\{ \bx\in\{0,1\}^n :{\rm wt}(\bx)=i \Big\} \to \Big\{ \bx\in \{0,1\}^n :{\rm wt}(\bx)=n-i \Big\},
\end{equation*}
\end{small}
and $\phi_i$ can be constructed in linear-time for all $n,i$ (refer to \cite{ROTH:ISIT}).


\subsection{Literature works on 2D $\epsilon$-balanced RC constrained codes ${\rm Bal}_{\rm RC}(n;\epsilon)$}

Over 1D codes, to encode binary balanced codes, we have the celebrated {\em Knuth's balancing technique} \cite{knuth:1986}. Knuth's balancing technique is a linear-time algorithm that maps a binary message $\bx$
to a balanced word $\bz$ of the same length by flipping the first $t$ bits of $\bx$.
The crucial observation demonstrated by Knuth is that such an index $t$ always exists and 
$t$ is commonly referred to as a {\em balancing index} of $\bx$. To represent such a balancing index, Knuth appends $\bz$ with a short balanced suffix of length $\ceil{\log n}$, which is also the redundancy of the encoding algorithm. Formally, we have the following theorem.
\begin{theorem}[Knuth \cite{knuth:1986} for $\epsilon=0$]
For arbitrary binary sequence $\bx\in\{0,1\}^n$, there exists the index $t$, where $1\leq t\leq n$, called a balancing index of $\bx$ such that the sequence $\by$ obtained by flipping the first $t$ bits in $\bx$, denoted by ${\rm Flip}_t(\bx)$, is balanced. 
There exists a pair of linear-time algorithms $\enc^{K}:\{0,1\}^{k}\to \{0,1\}^n$ and 
$\dec^{K}:\{0,1\}^n\to\{0,1\}^{k}$, where $k\approx n-\ceil{\log n}$, such that the following holds.
If $\enc^{K}(\bx)$ is balanced and $\dec^K\circ\enc^K(\bx)=\bx$ for all $\bx\in \{0,1\}^{k}$.
\end{theorem}

Modifications of the generic scheme are discussed for constructing more efficient balanced codes \cite{immink:2010,bose:1996} and almost-balanced codes \cite{TT:2020,TT:DNA,alon:1988}. Particularly, for 1D $\epsilon$-balanced codes, the encoding methods in \cite{TT:DNA,alon:1988} used only a constant redundant bits. Crucial to the improvement in redundancy from $\ceil{\log n}$ bit of Knuth's method for balanced codes to a constant number of bits for $\epsilon$-balanced codes is the construction of {\em $\epsilon$-balancing set}.

\begin{definition} For $n$ even, $\epsilon>0, \epsilon n\ge 1$, the index $t$, where $1\leq t\leq n$, is called an {\em $\epsilon$-balancing index} of $\bx\in\{0,1\}^n$ if $\by={\rm Flip}_t(\bx)$, is {\em $\epsilon$-balanced}. Suppose that $\epsilon n>1$, let the {\em $\epsilon$-balancing set} ${\rm S}_{\epsilon,n} \subset \{0,1,2,\ldots,n\}$ be the set of indices, given by ${\rm S}_{\epsilon,n}= \{0,n\} \cup \{2\floor{\epsilon n}, 4\floor{\epsilon n}, 6\floor{\epsilon n}, \ldots \}$. 
The size of ${\rm S}_{\epsilon,n}$ is at most $\floor{1/2\epsilon}+1$. 
\end{definition}

\begin{theorem}[Nguyen \et{} \cite{TT:DNA}]\label{epsilon-balanced}
Let $n$ be even, and $\epsilon>0,\epsilon n\ge1$. For arbitrary binary sequence $\bx \in \{0,1\}^{n}$, there exists an index $t$ in the set ${\rm S}_{\epsilon,n}$, such that $t$ is an $\epsilon$-balancing index of $\bx$. There exists a pair of linear-time algorithms $\enc^{\epsilon}:\{0,1\}^{k}\to \{0,1\}^n$ and 
$\dec^{\epsilon}:\{0,1\}^n\to\{0,1\}^{k}$, where $k= n-2\log (\floor{1/2\epsilon}+1)$, such that the following holds.
If $\enc^{\epsilon}(\bx)$ is $\epsilon$-balanced and $\dec^{\epsilon}\circ\enc^{\epsilon}(\bx)=\bx$ for all $\bx\in \{0,1\}^{k}$. The redundancy of the encoder is $2\log (\floor{1/2\epsilon}+1)=\Theta(1)$. 
\end{theorem}


\begin{example}
Consider $n=10, \epsilon=0.1$, we then have ${\rm S}_{\epsilon,n}=\{0,2,4,6,8\}$. Let $\bx=0000000000$.
Observe that $f_4(\bx)={\tt {\color{red}1111}000000}$, and $f_6(\bx)={\tt {\color{red}111111}0000}$ are $\epsilon$-balanced. Hence, $t=4,6$ are $\epsilon$-balancing indices of $\bx$. In general, there might be more than one $\epsilon$-balancing index.
\end{example}

For 2D $\epsilon$-balanced RC codes ${\rm Bal}_{\rm RC}(n;\epsilon)$ and $\epsilon=0$, Talyansky \et{} \cite{roth:1999} presented an efficient encoding method for $n \times n$ array, where each row and each column is balanced. The method uses $2n\log n+\Theta(n\log \log n)$ redundant bits and includes three phases. 
\begin{itemize}
\item In phase I, all rows in the array are encoded to be balanced via the Knuth's balancing technique, which uses roughly $\log n$ redundant bits for each row.  
\item In phase II, the array is then divided into two subarrays of equal size, and the problem is reduced to balancing two subarrays from a given balanced array. Here, an array is defined to be balanced when it contains equal number of 0's and 1's. The process is repeated until each subarray is as a single column. For each subproblem, to balance two subarrays, the authors simply {\em swap} the elements between two subarrays until both subarrays are balanced. The number of swapped elements, called {\em index}, will be encoded and decoded in phase III in order to recover the information. 
\item In phase III, the authors presented a method to encode/decode all the indices used in phase II recursively. The authors also proved that the redundancy in this phase is at most $n\log n+O(\log \log n)$ bits. 
\end{itemize}

To further reduce the redundancy, instead of using one of the algorithms in \cite{knuth:1986,bose:1996}, Talyansky \et{} balanced the rows with the (more computationally complex) enumerative coding technique. Consequently, the redundancy can be reduced to $1.5n\log n+\Theta(n\log \log n)$ redundant bits. 

In this work, we first extend the study in the literature works to design 2D $\epsilon$-balanced RC codes ${\rm Bal}_{\rm RC}(n;\epsilon)$ for arbitrary $\epsilon>0$. Particularly, for $n>0,\epsilon\in(0,1/2)$, we design efficient coding methods that encode (decode) binary data to (from) ${\rm Bal}_{\rm RC}(n;\epsilon)$. Clearly, since ${\rm Bal}_{\rm RC}(n;0) \subset {\rm Bal}_{\rm RC}(n;\epsilon)$ for all $\epsilon>0$, one may use the constructions of Talyansky \et{} \cite{roth:1999} that use  at most $1.5n\log n+\Theta(n\log \log n)$ redundant bits. Throughout this work, we show that the redundancy can be reduced significantly to be at most $\Theta(n)$ bits. Furthermore, for sufficiently large $n$, we show that there exist linear-time encoding (and decoding respectively) algorithms for ${\rm Bal}_{\rm RC}(n;\epsilon)$ that use only one redundant bit. 

\subsection{Our Contribution}

In this work, we present efficient encoding/decoding methods for 2D RC constrained codes and 2D subarray constrained codes for arbitrary $p\in(0,1)$ and $\epsilon\in(0,1/2)$. Method A is based on the divide and conquer algorithm and a modification of the Knuth's balancing technique while method B is based on the sequence replacement technique. We present a summary of the encoders and decoders proposed by this work in Table I.
\begin{enumerate}[(A)]
\item In Section III, we apply method A to encode ${\rm B}_{\rm RC}(n;p)$ when $p\le 1/2$ with at most $n \mu(n,p)+O(n+\log n)$ redundant bits (compared to $2n \mu(n,p)+O(n+\log n)$ if using the method in \cite{ROTH:2000}), and to encode ${\rm Bal}_{\rm RC}(n;\epsilon)$ for arbitrary $\epsilon\in(0,1/2)$ with at most $\Theta(n)$ redundant bits (compared to $1.5n\log n+\Theta(n\log \log n)$ if using the method \cite{roth:1999}). 
\item In Section IV, we first use method B to encode ${\rm B}_{\rm RC}(n;p)$ when $p>1/2$ with at most $n+3$ redundant bits. We then show that for sufficiently large $n$, we can encode ${\rm Bal}_{\rm RC}(n;\epsilon)$ or encode ${\rm B}_{\rm RC}(n;p)$ when $p>1/2$ with only one redundant bit.
\item In Section V, we propose encoding/decoding methods for 2D subarray constrained codes ${\rm B}_{\rm S}(n,m;p)$ and ${\rm Bal}_{\rm S}(n,m;\epsilon)$. 
\end{enumerate}

\begin{table}
\renewcommand{\arraystretch}{1.3}
\begin{tabular}{ p{3.65cm} p{7cm} p{3.25 cm} p{3.15cm}}
\hline
  Encoder / Decoder   & Description & Redundancy & Input Requirement \\
 \hline
  
  
  
  
  
  
  $\enc_{{\rm B}_{\rm RC}(n;p)}^{1}, \dec_{{\rm B}_{\rm RC}(n;p)}^{1}$& Encoder and decoder for ${\rm B}_{\rm RC}(n;p)$: the weight of every row and every column is at most $pn$ when $p< 1/2$ using the divide and conquer method & 
         $n\lambda(n,p)+ O(n\log n)$ bits & $n>1/p (1+\log n)$\\
                   $\enc_{{\rm B}_{\rm RC}(n;p)}^{2}, \dec_{{\rm B}_{\rm RC}(n;p)}^{2}$& Encoder and decoder for ${\rm B}_{\rm RC}(n;p)$: the weight of every row and every column is at most $pn$ when $p> 1/2$ using the sequence replacement technique & 
         $(n+3)$ bits & $n\ge 1/c^2 \ln (n^2-n-3)$ where $c=p-1/2$\\
      \hline     
           $\enc_{{\rm Bal}_{\rm RC}(n;\epsilon)}^{1}, \dec_{{\rm Bal}_{\rm RC}(n;\epsilon)}^{1}$& Encoder and decoder for ${\rm Bal}_{\rm RC}(n;\epsilon)$: the weight of every row and every column is within $[(1/2-\epsilon)n,(1/2+\epsilon)n]$ using the divide and conquer method & 
         $3cn-2c^2=\Theta(n)$ bits, and $c=2\ceil{\log(\floor{1/2\epsilon}+1)} $& $n>2c, n\epsilon \ge 1$\\
                   $\enc_{{\rm Bal}_{\rm RC}(n;\epsilon)}^{2}, \dec_{{\rm Bal}_{\rm RC}(n;\epsilon)}^{2}$& Encoder and decoder for ${\rm Bal}_{\rm RC}(n;\epsilon)$: the weight of every row and every column is within $[(1/2-\epsilon)n,(1/2+\epsilon)n]$ using the sequence replacement technique & 
        ${\color{blue}{1}}$ {\color{blue}{bits}} & $n^2 \ge 8/\epsilon^2 \ln n, n\epsilon \ge 2$\\
          \hline     
           $\enc_{{\rm B}_{\rm S}(n,m;p)}, \dec_{{\rm B}_{\rm S}(n,m;p)}$& Encoder and decoder for ${\rm B}_{\rm S}(n,m;p)$: the weight of every subarray of size $m\times m$ is at most $pm^2$ where $m=n-k$, and $p\ge 1/2$ using the antipodal matching& 
         $2(k+1)^2$ bits & $k=o(n)$\\
                   $\enc_{{\rm Bal}_{\rm S}(n,m;\epsilon)}, \dec_{{\rm Bal}_{\rm S}(n,m;\epsilon)}$& Encoder and decoder for ${\rm Bal}_{\rm S}(n,m;\epsilon)$: the weight of every subarray of size $m\times m$ is within $[(1/2-\epsilon)m^2,(1/2+\epsilon)m^2]$ using the sequence replacement technique & 
        ${\color{blue}{1}}$ {\color{blue}{bits}} & $2/\epsilon^2 \ln n \le m \le n$\\

 \hline 
 \end{tabular}
\caption{A summary of the encoders and decoders proposed by this work. The redundancy is computed for array codewords of length $n\times n$, given $p, \epsilon>0$. }\label{tab.notation}
\end{table}




\section{Efficient Encoders/Decoders for ${\rm B}_{\rm RC}(n;p)$ and ${\rm Bal}_{\rm RC}(n;\epsilon)$ via The Divide and Conquer Algorithm}
We first define the swapping function of two binary sequences. 
\begin{definition}
Given $\by=y_1y_2\ldots y_m, \bz=z_1z_2\ldots z_m$. For $1\le t\le m$, we use ${\rm Swap}_t(\by,\bz)$, ${\rm Swap}_t(\bz,\by)$ to denote the sequences obtained by swapping the first $t$ bits of $\by$ and $\bz$, i.e.
\begin{small}
\begin{align*}
{\rm Swap}_t(\by,\bz) &= z_1z_2\ldots z_t y_{t+1}y_{t+2}\ldots y_m, \text{ and } \\
{\rm Swap}_t(\bz,\by) &= y_1y_2\ldots y_t z_{t+1}z_{t+2}\ldots z_m. 
\end{align*}
\end{small}
\end{definition}

A key ingredient of the encoding method in \cite{roth:1999} is the following lemma. 
\begin{lemma}[Swapping Lemma]\label{swap-element}
Given  $\bx=\by\bz \in \{0,1\}^{2m}$, $\bx$ is balanced, $\by=y_1y_2\ldots y_m, \bz=z_1z_2\ldots z_m$. There exists an index $t$, $1\le t\le m$ such that both ${\rm Swap}_t(\by,\bz)$ and ${\rm Swap}_t(\bz,\by)$ are balanced. 
\end{lemma}

Lemma~\ref{swap-element} states that there must be an index $t$, referred as a {\em swapping index} of $\by$ and $\bz$, such that after swapping their first $t$ bits, the resulting sequences are both balanced. Since such an index $t$ belongs to $\{1,2,\ldots m\}$, in order to recover the original sequences $\by,\bz$ from the output sequences ${\rm Swap}_t(\by,\bz)$ and ${\rm Swap}_t(\bz,\by)$, a redundancy of $\log m$ (bits) is needed. Therefore, to encode binary array of size $n\times n$, the method in \cite{roth:1999} used a total of $n\log n+O(\log \log n)$ redundant bits to store all swapping indices of all pairs of subarrays. It is easy to verify that Lemma~\ref{swap-element} also works for $p$-bounded constraint and $\epsilon$-balanced constraint for all $p,\epsilon>0$. 

\begin{corollary}\label{coro1}
Given  $\bx=\by\bz \in \{0,1\}^{2m}$, $\by=y_1y_2\ldots y_m, \bz=z_1z_2\ldots z_m$. If $\bx$ is $p$-bounded then there exists an index $t$, $1\le t\le m$ such that both ${\rm Swap}_t(\by,\bz)$ and ${\rm Swap}_t(\bz,\by)$ are $p$-bounded. Similarly, if $\bx$ is $\epsilon$-balanced then there exists an index $t$, $1\le t\le m$ such that both ${\rm Swap}_t(\by,\bz)$ and ${\rm Swap}_t(\bz,\by)$ are $\epsilon$-balanced.
\end{corollary} 

\subsection{Design of ${\rm B}_{\rm RC}(n;p)$ when $p<1/2$}

We now present efficient encoding method for 2D $p$-bounded RC constrained codes ${\rm B}_{\rm RC}(n;p)$ for given $p<1/2$. Recall that one may follow the coding method in \cite{ROTH:2000}, based on enumeration coding, that ensure the weight in every row and every column to be precisely $pn$. The redundancy of the proposed encoder was at most $2n \mu(n,p)+O(n+\log n)$, where $\mu(n,p)$ is the least redundancy required to encode one-dimensional binary codewords of length $n$ such that the weight is $pn$ (refer to Section II-B). 

In this section, we adapt the divide and conquer algorithm with the enumeration coding technique. Compare to literature works in \cite{ROTH:2000,C:2020} that also used modifications of enumeration coding technique, the major difference of our coding method is that the rows and columns are encoded independently, and the 2D code construction can be divided to two 1D code constructions. In other words, the complexity of our encoder mainly depends on the efficiency of enumeration coding for 1D codes. 

In general, a {\em ranking function} for a finite set $S$ of cardinality $N$ is a bijection ${\rm rank} : {\bf S} \to [N]$ where $[N] \triangleq \{0,1,2,\ldots,N-1\}$. Associated with the function rank is a unique {\em unranking function} ${\rm unrank} : [N] \to {\bf S}$, such that ${\rm rank}(s) = j$ if and only if ${\rm unrank}(j) = s$ for all $s \in {\bf S}$ and $j \in [N]$. Given $n,p>0$, let ${\bf S}(n,p)\triangleq \Big\{ \bx \in \{0,1\}^n: {\rm wt}(\bx) \le pn \Big\}$, i.e. ${\bf S}(n,p)$ is the set of all 1D sequences that satisfy the $p$-bounded constraint. One may use enumeration coding \cite{book,cover:1973} to construct ${\rm rank}_p : {\bf S}(n,p) \to [|{\bf S}(n,p)|]$ and ${\rm unrank}_p :  [|{\bf S}(n,p)|] \to {\bf S}(n,p)$. The redundancy of this encoding algorithm is then $\lambda(n,p)=n-\log|{\bf S}(n,p)|$ (bits). 

It is easy to verify that $\lambda(n,p)\le\mu(n,p)$. The redundancy of our encoder is at most $n\lambda(n,p)+ O(n\log n) \le n\mu(n,p)+O(n\log n)$ bits. We now describe the main idea of the algorithm. The encoding method includes three phases.
\vspace{0.05in}

\noindent {\bf Encoder} $\enc_{{\rm B}_{\rm RC}(n;p)}^{1}$. Set $c=1/p(1+\log n)$. The binary data $\bx$ is of length $N= (n-c)(n-\lambda(n,p))$, i.e. the redundancy is then $n\lambda(n,p)+cn+c\lambda(n,p)=n\lambda(n,p)+O(n\log n)$ bits. 

\begin{itemize}
\item In Phase I, the encoder encodes the information of length $N$ into array $A$ of size $(n-c) \times n$ where every row is $p$-bounded, using the enumeration coding technique. Particularly, the information is encoded into ${\bf S}(n,p)$ with the redundancy at most $\lambda(n,p)$ bits for each row. Therefore, the redundancy used in Phase I is at most $(n-c)\lambda(n,p) < n\lambda(n,p)$ bits. 

\item In Phase II, the encoder ensures that every column of $A$ is $p$-bounded. Note that, from Phase I, array $A$ is $p$-bounded. Here, an array of size $M$ is said to be $p$-bounded if the number of ones is at most $pM$. Suppose that at some encoding step $i$, we have an array $S$ of size $(n-c) \times n_0$, which is already $p$-bounded for some $n_0\le n$ (initially, $n_0\equiv n$). We then divide $S$ into two subarrays of size $(n-c)\times (n_0/2)$, denoted by $L_S$ and $R_S$, and proceed to ensure that $L_S$ and $R_S$ are both $p$-bounded. The encoder follows Corollary~\ref{coro1} to find a swapping index so that $L_S$ and $R_S$ are both $p$-bounded. To represent such a swapping index $t$, we need at most $\log (n-c)n_0$ redundant bits. After both subarrays $L_S$ and $R_S$ are $p$-bounded, we continue to divide each of them into two subarrays and repeat the process to ensure that the newly created subarrays are also $p$-bounded. This process ends when all subarrays of size $(n-c) \times 1$ are $p$-bounded. We illustrate the idea of the swapping procedure in Figure 1. Let ${\rm Re}(n)$ be the sequence obtained by concatenating all binary representations of all swapping indices. The size of ${\rm Re}(n)$, is at most 
\begin{small}
\begin{align*}
\sum_{{\substack{k=2^j\\
                            2\le k\le n}}}
 (n/k) \log ((n-c)k/2)&=(n-1)(1+\log (n-c))-\log n\\
 &<n(1+\log n) = O(n\log n) \text{  (bits)}.
\end{align*}
\end{small}

\item In Phase III, the encoder encodes ${\rm Re}(n)$ into an array $B$ of size $c \times n$ such that its every row and every column is $p$-bounded. At the end of Phase III, the encoder outputs the concatenation of $A$ and $B$, which is an array of size $n \times n$. Recall that $c=(1/p) (1+\log n)$. In contrast to the encoder in \cite{roth:1999}, which proceeds to repeat the encoding procedure to encode ${\rm Re}(n)$, we show that ${\rm Re}(n)$ can be encoded/decoded efficiently without repeating the encoding procedure. Suppose that ${\rm Re}(n)=x_1x_2\ldots x_{n(1+\log n)}$, ${\rm Re}(n)$ is encoded to $B$ as follows: $B_1 = (x_1 0^{\frac{1}{p}-1}) (x_{c+1} 0^{\frac{1}{p}-1}) \ldots (x_{c(n-1)+1} 0^{\frac{1}{p}-1})$, $B_i = (0^{j-1} x_{i} 0^{\frac{1}{p}-j})  \ldots  (0^{j-1} x_{c(n-1)+i} 0^{\frac{1}{p}-j})$, where $j\equiv i \ppmod{1/p}$, $1\le i\le n-m$ and $0^t$ denotes the repetition of $0$ for $t$ times. 
It is easy to verify that every row and every column of $B$ is $p$-bounded and ${\rm Re}(n)$ can be decoded uniquely from $B$. 
\end{itemize}

In conclusion, the total redundancy for encoding an $n\times n$ array is bounded above by
\begin{align*}
(n-c)\lambda(n,p)+(1/p) n (1+\log n) <  n \lambda(n,p)+O(n\log n) < n\mu(n,p)+ O(n\log n) \text{ bits}.
\end{align*}
In general, when $pn$, $\log n$ and $1/p$ are not integers, we can easily modify the construction of our encoder and show that the redundancy is at most $n\lambda(n,p)+ O(n\log n)$ bits. For completeness, we present the corresponding decoding algorithm as follows. 

\begin{figure}[!t]
	\begin{itemize}
	
	\item [(a)] We divide $A$ into two subarrays $L_A$ and $R_A$. In this step, only $L_A$ is not $p$-bounded. To obtain two $p$-bounded subarrays, we swap their prefixes of length {\color{teal}{six}}. In other words, we set $L_A$ and $R_A$ to be ${\rm Swap}_2(L_A,R_A)$ and ${\rm Swap}_2(R_A,L_A)$, respectively.
	
	\vspace{3mm}
	
	\begin{center}
	\begin{tikzpicture}
		
		\arraycolsep=3pt
		\tikzset{unb/.style = {rectangle, draw=red, line width=2pt, align=center, inner sep = 1.5pt}}
		\tikzset{bal/.style = {rectangle, draw=blue, line width=2pt, align=center, inner sep = 1.5pt}}
		\tikzset{arrow/.style = {->,> = {Stealth[scale=1.5]},black,line width=2pt}}
		
		\node[unb](LS) at (0,0) {
			$\begin{array}{cccc}
			1 & 0 & 0 & 1 \\
			0 & 1 & 1 & 0 \\
			0 & 1 & 0 & 0 \\
			0 & 1 & 0 & 0 \\ 	
			\end{array}$};
		\node[bal](RS) at (1.7,0) {
			$\begin{array}{cccc}
				0 & 0 & 0 & 0 \\
				0 & 0 & 0 & 0 \\
				0 & 0 & 1 & 0 \\
				0 & 0 & 0 & 1 \\ 	
			\end{array}$};
		
		\node[bal](bLS) at (4.5,0) {
			$\begin{array}{cccc}
				{\color{teal}\bf 0} & {\color{teal}\bf 0} & 0 & 1 \\
				{\color{teal}\bf 0} & {\color{teal}\bf 0} & 1 & 0 \\
				{\color{teal}\bf 0} & 1 & 0 & 0 \\
				{\color{teal}\bf 0} & 1 & 0 & 0 \\ 	
			\end{array}$};
		\node[bal](bRS) at (6.2,0) {
			$\begin{array}{cccc}
				{\color{orange}\bf 1} & {\color{orange}\bf 0} & 0 & 0 \\
				{\color{orange}\bf 0} & {\color{orange}\bf 1} & 0 & 0 \\
				{\color{orange}\bf 0} & 1 & 0 & 0 \\
				{\color{orange}\bf 0} & 0 & 0 & 1 \\ 	
			\end{array}$};
		
		\draw[arrow](RS.east) to (bLS.west);
		\node at (0,-1.2) {\scriptsize$L_A$};
		\node at (1.7,-1.2) {\scriptsize$R_A$};
		\node at (4.3,-1.2) {\scriptsize${\rm Swap}_2(L_A,R_A)$};
		\node at (6.3,-1.2) {\scriptsize${\rm Swap}_2(R_A,L_A)$};
	\end{tikzpicture}
	\end{center}

	\item [(b)] We divide $L_A$ into two subarrays $L_{L_A}$ and $R_{L_A}$. Here, both subarrays are $p$-balanced and there is no swapping of prefixes. When we divide $R_A$ into two subarrays $L_{R_A}$ and $R_{R_A}$, we observe that $L_{R_A}$ is not $p$-bounded. We proceed as before and swap their prefixes of length {\color{teal}one}. 
	\vspace{3mm}
	
	\begin{center}
		\begin{tikzpicture}
			
			\arraycolsep=2pt
			\tikzset{unb/.style = {rectangle, draw=red, line width=2pt, align=center, inner sep = 1.5pt}}
			\tikzset{bal/.style = {rectangle, draw=blue, line width=2pt, align=center, inner sep = 1.5pt}}
			\tikzset{arrow/.style = {->,> = {Stealth[scale=1.5]},black,line width=2pt}}
			
			\node[bal](LL) at (0,0) {
				$\begin{array}{cccc}
					0 & 0 \\
					0 & 0 \\
					0 & 1 \\
					0 & 1 \\ 	
				\end{array}$};
			\node[bal](RL) at (0.9,0) {
				$\begin{array}{cccc}
				  	0 & 1 \\
					1 & 0 \\
					0 & 0 \\
					0 & 0 \\ 	
				\end{array}$};
			\node[unb](LR) at (3,0) {
				$\begin{array}{cccc}
					1 & 0 \\
					0 & 1 \\
					0 & 1 \\
					0 & 0 \\ 	
				\end{array}$};
			\node[bal](LR) at (3.9,0) {
				$\begin{array}{cccc}
					0 & 0 \\
					0 & 0 \\
					0 & 0 \\
					0 & 1 \\ 	
				\end{array}$};
			\node at (0  ,1.2) {\scriptsize$L_{L_A}$};
			\node at (0.9,1.2) {\scriptsize$R_{L_A}$};
			\node at (3  ,1.2) {\scriptsize$L_{R_A}$};
			\node at (3.9,1.2) {\scriptsize$R_{R_A}$};
			
			\node[bal](bLL) at (0,-3) {
				$\begin{array}{cccc}
					0 & 0 \\
					0 & 0 \\
					0 & 1 \\
					0 & 1 \\ 	
				\end{array}$};
			\node[bal](bRL) at (0.9,-3) {
				$\begin{array}{cccc}
					0 & 1 \\
					1 & 0 \\
					0 & 0 \\
					0 & 0 \\ 	
				\end{array}$};
			\node[bal](bLR) at (3,-3) {
				$\begin{array}{cccc}
					{\color{teal}\bf 0} & 0 \\
					0 & 1 \\
					0 & 1 \\
					0 & 0 \\ 	
				\end{array}$};
			\node[bal](bLR) at (3.9,-3) {
				$\begin{array}{cccc}
					{\color{orange}\bf 1} & 0 \\
					0 & 0 \\
					0 & 0 \\
					0 & 1 \\ 	
				\end{array}$};
			\draw[arrow] (0.5,-1) to (0.5,-2);
			\draw[arrow] (3.5,-1) to (3.5,-2);
			\node at (-0.5,-1.5) {\footnotesize No swapping};
			\node[text width=2.5cm] at (5,-1.5) {\footnotesize Swapping of prefixes of length one};
		\end{tikzpicture}
	\end{center}
	
	\item [(c)] Finally, we divide each of the four subarrays into two, resulting in the $n=8$ columns. 
	The final output is as follows.
	\vspace{1mm}
	
	\begin{center}
		\begin{tikzpicture}
			
			\arraycolsep=2pt
			\tikzset{unb/.style = {rectangle, draw=red, line width=2pt, align=center, inner sep = 1.5pt}}
			\tikzset{bal/.style = {rectangle, draw=blue, line width=2pt, align=center, inner sep = 1.5pt}}
			\tikzset{arrow/.style = {->,> = {Stealth[scale=1.5]},black,line width=2pt}}
			
			\node[bal](c1) at (0,0) {
				$\begin{array}{cccc}
					{\color{teal}\bf 0} \\
					{\color{teal}\bf 0} \\
					{\color{teal}\bf 1} \\
					0 \\ 	
				\end{array}$};
			\node[bal](c2) at (0.6,0) {
				$\begin{array}{cccc}
					{\color{orange}\bf 0} \\
					{\color{orange}\bf 0} \\
					{\color{orange}\bf 0} \\
					1 \\ 	
				\end{array}$};
			\node[bal](c3) at (1.5,0) {
				$\begin{array}{cccc}
					0 \\
					1 \\
					0 \\
					0 \\ 	
				\end{array}$};
			\node[bal](c4) at (2.1,0) {
				$\begin{array}{cccc}
					1 \\
					0 \\
					0 \\
					0 \\ 	
				\end{array}$};
			
			\node[bal](c5) at (3.2,0) {
				$\begin{array}{cccc}
					{\color{teal}\bf 0}  \\
					{\color{teal}\bf 1}  \\
					0  \\
					0  \\ 	
				\end{array}$};
			\node[bal](c6) at (3.8,0) {
				$\begin{array}{cccc}
					{\color{orange}\bf 0}  \\
					{\color{orange}\bf 0}  \\
					1  \\
					0  \\ 	
				\end{array}$};

			\node[bal](c7) at (4.7,0) {
				$\begin{array}{cccc}
					1 \\
					0 \\
					0 \\
					0 \\ 	
				\end{array}$};
			\node[bal](c8) at (5.3,0) {
			$\begin{array}{cccc}
			 0 \\
			 0 \\
			 0 \\
			 1 \\ 	
			\end{array}$};
			
		\end{tikzpicture}
	\end{center}

	\end{itemize}
	
	\caption{Example for $n=8, p=1/4$. The current subarray $A$ is of size $8\times 4$. The subarrays, highlighted in {\color{red}red}, are not $p$-bounded while those, highlighted in {\color{blue}blue}, are $p$-bounded.
	}
	\label{fig1}
\end{figure}
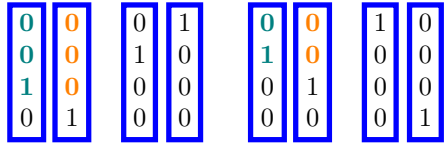


\vspace{0.05in}
\noindent{\bf Decoder, $\dec_{{\rm B}_{\rm RC}(n;p)}^{1}$}. 
\vspace{0.05in}

{\sc Input}: $A \in {\rm B}_{\rm RC}(n;p)$ of size $n\times n$\\
{\sc Output}: $\bx \triangleq \dec_{{\rm B}_{\rm RC}^1(n;p)}(A) \in \{0,1\}^N$, where $N=(n-c)(n-\lambda(n,p))$, $c=1/p(1+\log n)$\\[-2mm]

\begin{enumerate}[(I)]
\item Let $B$ be the subarray obtained by the last $c$ rows of $A$ and $C$ be the subarray obtained by the first $(n-c)$ rows of $A$
\item Decode ${\rm Re}(n)$ of length $n(1+\log n)$ from $B$ 
\item Do the reverse swapping process in $C$ according to ${\rm Re}(n)$, $C$ is an array of size $(n-c) \times n$. Let $\by_i$ be the $i$th row of $C$, we then have $\by_i \in {\bf S}(n,p)$ for all $1\le i\le n-c$
\item For $1\le i\le n-c$, let $\bz_i$ be the binary sequence of length $(n-\lambda(n,p))$ representing ${\rm rank}(\by_i)$
\item Output $\bx\triangleq \bz_1\bz_2\ldots \bz_{n-c} \in \{0,1\}^{(n-c)(n-\lambda(n,p))}$
\end{enumerate}
\vspace{0.05in}

\subsection{Design of ${\rm Bal}_{\rm RC}(n;\epsilon)$}

For $\epsilon$-balanced constraint, similar to what we achieved in Theorem~\ref{epsilon-balanced}, a swapping index of two sequences of size $m$ can be found in a set of constant size, which only depends on $\epsilon$. Consequently, the redundancy to encode each swapping index is reduced from $\log m$ bits to a constant number of bits. Recall the construction of the $\epsilon$-balancing set ${\rm S}_{\epsilon,n} \subset \{0,1,2,\ldots,n\}$, which is the set of indices of size at most $\floor{1/2\epsilon}+1$,
\begin{equation*}
{\rm S}_{\epsilon,n}= \{0,n\} \cup \{2\floor{\epsilon n}, 4\floor{\epsilon n}, 6\floor{\epsilon n}, \ldots \}.
\end{equation*}

\begin{lemma}[Modified Swapping Lemma]\label{swap-modify}
Given  $\bx=\by\bz \in \{0,1\}^{2m}$, $\bx$ is $\epsilon$-balanced, $\by=y_1y_2\ldots y_m, \bz=z_1z_2\ldots z_m$. There exists an index $t\in {\rm S}_{\epsilon,m}$, referred as a swapping $\epsilon$-balanced index of $\by$ and $\bz$, such that both ${\rm Swap}_{t}(\by,\bz)$ and ${\rm Swap}_{t}(\bz,\by)$ are $\epsilon$-balanced. 
\end{lemma}

\begin{proof}
Assume that ${\rm wt}(\by)<(1/2-\epsilon)m$ while ${\rm wt}(\bz)>(1/2+\epsilon)m$. Let $t_k=2k\floor{\epsilon m},$ for $k\ge 1, t_k \in {\rm S}_{\epsilon,m}$. Observe that, for any $i>0$, the weights of ${\rm Swap}_{t_i}(\by,\bz)$ and  ${\rm Swap}_{t_{i+1}}(\by,\bz)$ differ at at most $2\floor{\epsilon m}$. The same argument applies for ${\rm Swap}_{t_i}(\bz,\by)$ and  ${\rm Swap}_{t_{i+1}}(\bz,\by)$. 
 
In addition, when $t=m$, we have ${\rm wt}\big({\rm Swap}_{m}(\by,\bz)\big)>(1/2+\epsilon)m$ while ${\rm Swap}_{m}(\bz,\by)\big)<(1/2-\epsilon)m$. There must be an index $t\in {\rm S}_{\epsilon,m}$ such that ${\rm wt}\big({\rm Swap}_{t}(\by,\bz)\big), {\rm wt}\big({\rm Swap}_{t}(\bz,\by)\big) \in \big[(1/2-\epsilon)m,(1/2+\epsilon)m\big]$.
\end{proof}

Note that the swapping $\epsilon$-balancing index might not be unique. Besides, the set ${\rm S}_{\epsilon,m}$ might not include all the possible swapping $\epsilon$-balanced indices. For simplicity, during the encoding process, we simply use the smallest possible index in ${\rm S}_{\epsilon,m}$. To represent such an index, a binary sequence of size $\log\big|{\rm S}_{\epsilon,m} \big|=\ceil{\log \big(\floor{1/2\epsilon}+1\big)}$ is needed.
\begin{example}
Consider $m=10, \epsilon=0.1$, $\by=1110000000$ and $\bz=1101101111$. We then have ${\rm S}_{\epsilon,m}=\{2,4,6,8,10\}$. We can verify that $t=6,8 \in {\rm S}_{\epsilon,m}$ are both swapping $\epsilon$-balanced index. In addition $t=5 \notin {\rm S}_{\epsilon,m}$ is also a swapping $\epsilon$-balanced index.
\end{example}

\noindent {\bf Encoder} $\enc_{{\rm Bal}_{\rm RC}(n;\epsilon)}^{1}$. For simplicity, we assume that $\log n$ is integer. Set $c=2\ceil{\log \big(\floor{1/2\epsilon}+1\big)}$. The binary data $\bx$ is of length $N= (n-2c)(n-c)=n^2-3cn+2c^2$, i.e. the redundancy is then $3cn-2c^2=\Theta(n)$ bits. Similar to the construction of the encoder $\enc_{{\rm B}_{\rm RC}(n;p)}^{1}$ in Subsection III-A, the encoding algorithm for $\enc_{{\rm Bal}_{\rm RC}(n;\epsilon)}^{1}$ includes three phases.

\begin{itemize}
\item In phase I. All rows in the array are encoded to be $\epsilon$-balanced via the encoding method in \cite{TT:DNA,alon:1988}, that uses at most $c=\Theta(1)$ redundant bits for each row. Particularly, the data is first divided into $(n-2c)$ sequences of equal length $n-c$ bits. Suppose $\bx=\bx_1\bx_2\ldots \bx_{n-2c}$ where $\bx_i\in \{0,1\}^{n-c}$. For $1\le i\le n-2c$, the encoder uses the method in \cite{TT:DNA,alon:1988} to encode $\bx_i$ to obtain the row $A_i$ of length $n$ and $A_i$ is $\epsilon$-balanced. 

\item  In phase II, the array is then divided into two subarrays of equal size, and the problem is reduced to obtain two $\epsilon$-balanced subarrays from a given $\epsilon$-balanced array. Here, an array containing $m$ bits is defined to be $\epsilon$-balanced if the number of 1's is within $\big[(1/2-\epsilon)m,(1/2+\epsilon)m\big]$ (which is similar to Definition 1). The process is repeated until each subarray is as a single column. For each subproblem, to obtain two $\epsilon$-balanced subarrays, we follow Lemma~\ref{swap-modify} to find a swapping $\epsilon$-balanced index. To represent such an index, we need at most $\ceil{\log \big(\floor{1/2\epsilon}+1\big)}=c/2$ redundant bits. Let ${\rm Re}(n)$ be the sequence obtained by concatenating all binary representations of all swapping indices. The size of ${\rm Re}(n)$, is at most 
\begin{small}
\begin{align*}
\sum_{{\substack{k=2^j\\
                            2\le k\le n}}}
 (n/k) c/2&<cn/2 \text{  (bits)}.
\end{align*}
\end{small}

\item In Phase III, the encoder encodes ${\rm Re}(n)$ into an array $B$ of size $2c\times n$ such that its every row and every column is $\epsilon$-balanced. At the end of Phase III, the encoder outputs the concatenation of $A$ and $B$, which is an array of size $n \times n$. Similar to the construction of the encoder $\enc_{{\rm B}_{\rm RC}(n;p)}^{1}$ in Subsection III-A, we also show that ${\rm Re}(n)$ can be encoded/decoded efficiently without repeating the swapping procedure. Suppose that ${\rm Re}(n)=r_1r_2\ldots r_{cn/2}$, ${\rm Re}(n)$ is encoded to $B$ as follows. We first fill the information bits into all the odd rows of B as $B_1 = r_1 \overline{r_1} r_2 \overline{r_2} \ldots r_{n/2} \overline{r_{n/2}}, \ldots, B_{cn} = r_{(c-1)n/2+1} \overline{r_{(c-1)n/2}}\ldots r_{cn/2} \overline{r_{cn/2}}$. For the even row, we simply set $B_{2i}=\overline{B_{2i-1}}$ for $1\le i\le c$. It is easy to verify that every row and every column of $B$ is balanced and ${\rm Re}(n)$ can be decoded uniquely from $B$. 
\end{itemize}

Recall that if the encoder follows Corollary~\ref{coro1} to find a swapping $\epsilon$-balancing index for each step, the redundancy is at least $n\log n+O(n+\log n)$ bits. 
Similar to the construction of the decoder $\dec_{{\rm B}_{\rm RC}(n;p)}^{1}$, we have the construction of $\dec_{{\rm Bal}_{\rm RC}(n;\epsilon)}^{1}$.

\vspace{0.05in}
\noindent{\bf Decoder, $\dec_{{\rm Bal}_{\rm RC}(n;\epsilon)}^{1}$}. 
\vspace{0.05in}

{\sc Input}: $A \in {\rm Bal}_{\rm RC}(n;\epsilon)$ of size $n\times n$\\
{\sc Output}: $\bx \triangleq \dec_{{\rm Bal}_{\rm RC}(n;\epsilon)}^{1}(A) \in \{0,1\}^N$, where $N=(n-2c)(n-c)$, and $c=2\ceil{\log \big(\floor{1/2\epsilon}+1\big)}$\\[-2mm]

\begin{enumerate}[(I)]
\item Let $B$ be the subarray obtained by the last $2c$ rows of $A$ and $C$ be the subarray obtained by the first $(n-2c)$ rows of $A$
\item Decode ${\rm Re}(n)$ of length $cn/2$ from $B$ 
\item Do the reverse swapping process in $C$ according to ${\rm Re}(n)$, $C$ is an array of size $(n-2c) \times n$. Let $\by_i$ be the $i$th row of $C$, we then have $\by_i$ is $\epsilon$-balanced for all $1\le i\le n-2c$, $\by_i$ is of length $n$
\item For $1\le i\le n-2c$, let $\bz_i$ be the binary sequence of length $(n-c)$ after removing $c$ redundant bit from $\by_i$ (refer to $\epsilon$-balanced encoder/decoder for binary sequence in \cite{TT:DNA})
\item Output $\bx\triangleq \bz_1\bz_2\ldots \bz_{n-2c} \in \{0,1\}^{(n-2c)(n-c)}$
\end{enumerate}
\vspace{0.05in}

\section{Efficient Encoders/Decoders for ${\rm B}_{\rm RC}(n;p)$ and ${\rm Bal}_{\rm RC}(n;\epsilon)$ via The Sequence Replacement Technique}

The Sequence Replacement Technique (SRT) has been widely applied in the literature (for example, see \cite{TT:DNA,TT:2020,srt:2010,srt:2019}). This is an efficient method for removing forbidden substrings from a source word. The advantage of this technique is that the complexity of encoder and decoder is very low, and they are also suitable for parallel implementation. In general, the encoder removes the forbidden strings and subsequently inserts its representation (which also includes the position of the substring) at predefined positions in the sequence. In our recent work \cite{TT:2020}, for codewords of length $m$, we enforced the almost-balanced weight-constraint over every {\em window} of size $\ell=\Omega(\log m)$ (here, a window of size $\ell$ of $\bx$ refer $\ell$ consecutive bits in $\bx$). 
\begin{theorem}[Nguyen \et{} \cite{TT:2020}]\label{theoremSRT}
Given $p_1,p_2$ where $0\le p_1<1/2<p_2\le 1$, let $c=\min\{1/2-p_1,p_2-1/2\}$. For $(1/c^2) \ln m \le \ell \le m$, there exists linear-time algorithm $\enc :\{0,1\}^{m-1} \to\{0,1\}^m$ such that for all $\bx\in \{0,1\}^{m-1}$ if $\by=\enc(\bx)$ then ${\rm wt}(\by) \in [p_1m,p_2m])$ and for every window $\bw$ of size $\ell$ of $\by$, ${\rm wt}(\bw) \in [p_1\ell,p_2\ell]$.
\end{theorem}

In this section, we show that, for sufficiently large $n$, the redundancy to encode (decode) binary data to (from) ${\rm B}_{\rm RC}(n;p)$ when $p>1/2$ and ${\rm Bal}_{\rm RC}(n;\epsilon)$ for arbitrary $\epsilon\in(0,1/2)$ can be further reduced from $\Theta(n)$ bits to only a single bit via the SRT. Particularly, we provide two efficient encoders:
\begin{itemize}
\item The first encoder adapts the SRT (presented in \cite{TT:2020}) with the antipodal matching (constructed in \cite{ROTH:ISIT}) to encode arbitrary data to ${\rm B}_{\rm RC}(n;p)$ when $p>1/2$ with at most $n+3$ redundant bits. 
\item The second encoder, which is the main contribution of this work, modifies the SRT as presented in \cite{TT:2020} to encode ${\rm Bal}_{\rm RC}(n;\epsilon)$ with only one redundant bit. Since ${\rm Bal}_{\rm RC}(n;\epsilon) \subset {\rm B}_{\rm RC}(n;p)$ for all $\epsilon\le p-1/2$, this can be apply to encode ${\rm B}_{\rm RC}(n;p)$ when $p>1/2$ with one redundant bit as well. 
\end{itemize}

\subsection{SRT and Antipodal Matching for ${\rm B}_{\rm RC}(n;p)$ when $p>1/2$}

Recall that the encoders proposed in \cite{ROTH:ISIT} can be used for constructing ${\rm B}_{\rm RC}(n;p)$ when $p>1/2$, and the redundancy is roughly $2n$ (bits). In this section, we provide a linear-time encoder for ${\rm B}_{\rm RC}(n;p)$ where $p>1/2$ with at most $(n+3)$ redundant bits. We have $p>1/2$ in all our descriptions in this part.

Recall that for an array $A$ of size $n\times n$, we use $A_i$ to denote the $i$th row of $A$ and $A^j$ to denote the $j$th column of $A$. In addition, we use $A_{i;\langle j \rangle}$ to denote the sequence obtained by taking the first $j$ entries of the row $A_i$ and use $A^{i;\langle j \rangle}$ to denote the sequence obtained by taking the first $j$ entries of the column $A^i$. For example, if 
 \[A= \left( \begin{array}{ccc}
a & b & c\\
d & e & f\\
g & h & i\end{array} \right),\text{ then } A_{1;\langle 2 \rangle}=ab, A^{3;\langle 2 \rangle}=cf.\]

We now describe the detailed construction of the Encoder $\enc_{{\rm B}_{\rm RC}(n;p)}^{2}$, where $p>1/2$.
\vspace{0.05in}

\noindent {\bf Encoder} $\enc_{{\rm B}_{\rm RC}(n;p)}^{2}$, where $p>1/2$. Set $N=n^2-(n+3)$, $\ell=n$, $p_1=0$, $p_2=p$ and $c=p-1/2$. According to Theorem~\ref{theoremSRT}, for sufficient $n$ that $(1/c^2) \ln (n^2-n-3) \le n \le n^2-n-3$, there exists linear-time encoder, $\enc_{\rm seq}:\{0,1\}^N \to \{0,1\}^{N+1}$ such that for all $\bx \in \{0,1\}^N$ and $\by=\enc_{\rm seq}(\bx)$ we have ${\rm wt}(\bw) \in [0,pn]$ for every window $\bw$ of size $n$ of $\by$. In addition, we follow \cite{ROTH:ISIT} to construct the antipodal matchings $\phi$ for sequences of length $n-1$. 
\vspace{0.05in}


{\sc Input}: $\bx \in \{0,1\}^N$\\
{\sc Output}: $A \triangleq \enc_{{\rm B}_{\rm RC}(n;p)}^{2}(\bx) \in {\rm B}_{\rm RC}(n;p)$ with $p>1/2$\\[-2mm]

\begin{enumerate}[(I)]
\item Set $\by=\enc_{\rm seq}(\bx) \in \{0,1\}^{N+1}$. Suppose that $\by=y_1y_2\ldots y_{n^2-n-2}$. 
\item Fill $n^2-n-1$ bits of $\by$ to $A$ row by row as follows. 
\begin{itemize} 
\item Set $A_i \triangleq y_{n(i-1)+1}\ldots y_{ni}$ for $1\le i\le n-2$.
\item Set $A_{n-1}\triangleq y_{n(n-2)+1}\ldots y_{n^2-n-2}*_1*_2$, where $*_1,*_2$ are determined later.
\item Suppose that $A_{n}=z_1z_2\ldots z_n$ where $z_i$ is determined later. 
\item If ${\rm wt}\Big(A_{n-1;\langle n-2 \rangle}\Big)>p(n-2)$, flip all bits in $A_{n-1;\langle n-2 \rangle}$ and set $*_1=1$, otherwise set $*_1=0$.
\end{itemize} 

\item For $1\le i\le(n-1)$, we check the $i$th column:
\begin{itemize}
\item if ${\rm wt}\Big(A^{i;\langle n-1 \rangle}\Big)>pn$, set $z_i=1$ and replace $A^{i;\langle n-1 \rangle}$ with $\phi\Big(A^{i;\langle n-1 \rangle}\Big)$
\item Otherwise, set $z_i=0$. 
\end{itemize}

\item Check the $n$th row: 
\begin{itemize}
\item If ${\rm wt}\Big(A_{n;\langle n-1 \rangle} \Big)>pn$, set $*_2=1$ and replace $A_{n;\langle n-1 \rangle}$ with $\phi\Big( A_{n;\langle n-1 \rangle} \Big)$
\item Otherwise, set $*_2=0$. 
\end{itemize}
 
\item Check the $n$th column: 
\begin{itemize}
\item If ${\rm wt}\Big(A^{n;\langle n-1 \rangle}\Big)>pn$, set $z_n=1$ and replace $A^{n;\langle n-1 \rangle}$ with $\phi\Big(A^{n,\langle n-1 \rangle}\Big)$.
\item Otherwise, set $z_n=0$. 
\end{itemize}

\item Output $A$.
\end{enumerate}

\begin{theorem}\label{constrained proof}
The Encoder $\enc_{{\rm B}_{\rm RC}(n;p)}^{2}$ is correct. In other words, $\enc_{{\rm B}_{\rm RC}(n;p)}^{2}(\bx)\in {\rm B}_{\rm RC}(n;p)$ with $p>1/2$ for all $\bx \in \{0,1\}^N$. The redundancy is $n+3$ (bits).
\end{theorem}
\begin{proof}
Let $A=\enc_{{\rm B}_{\rm RC}(n;p)}^{2}(\bx)$. We first show that the weight of every column of $A$ is at most $pn$. From Step (III) and Step (V), the encoder guarantees that the weights of $n$ columns are at most $pn$. Although there is a replacement in the $n$th row in Step (IV), it does not affect the weight of any column. Indeed, from the definition of an antipodal matching, if ${\rm wt}(\bx)>n/2$, then $\phi(\bx)$ has all its 1's in positions where $\bx$ has 1's and ${\rm wt}(\phi(\bx))\le n/2$. Therefore, whenever the encoder performs replacement step in any row (or respectively any column), it does not increase the weight of any column (or respectively any row). Therefore, we conclude that ${\rm wt}\Big(A^{i}\Big)\le pn$ for all $1\le i\le n$.

We now show that the ${\rm wt}\Big(A_{i}\Big)\le pn$ for all $1\le i\le n$. From Step (IV), we observe that the $n$th row satisfies the weight-constraint. As mentioned above, during Step (III) and Step (V), whenever the encoder performs replacement step in any column, it does not increase the weight of any row, i.e. the weight-constraint is preserved over the first $(n-2)$ rows that is guaranteed by the Encoder $\enc_{\rm seq}$ from Step (I). It remains to show that the $(n-1)$th row satisfies the weight constraint with the determined values of $*_1,*_2$. Indeed, from Step (II), if $*_1=0$, we have: 
$${\rm wt}\Big(A_{n-1}\Big)\le p(n-2)+1 = pn + (1-2p) \le pn.$$
Otherwise, 
$${\rm wt}\Big(A_{n-1}\Big)< (1-p)(n-2)+2 < pn \text{ for all } n>2/(2p-1).$$
In conclusion, we have $\enc_{{\rm B}_{\rm RC}(n;p)}^{2}(\bx)\in {\rm B}_{\rm RC}(n;p)$ for all $\bx \in \{0,1\}^m$. Since $N=n^2-(n+3)$, the redundancy of our encoder is then $n^2-N=n+3$ (bits).
\end{proof}

\begin{remark}
We now discuss the lower bound for $n$ so that the encoding algorithm works. It requires $(1/c^2) \log_e (n^2-n-3) \le n \le n^2-n-3$ where $c=p-1/2$. Since $n^2-n-3\ge n$ for all $n\ge 3$ and $\ln (n^2-n-3)\le \ln n^2=2\ln n$, a simple lower bound can be described as $n\ge 3$ and $n/\ln n \geq 1/(p-1/2)^2$. 
\end{remark} 

For completeness, we describe the corresponding decoder $\dec_{{\rm B}_{\rm RC}(n;p)}^{2}$ as follows. 

\vspace{0.05in}
\noindent{\bf Decoder, $\dec_{{\rm B}_{\rm RC}(n;p)}^{2}$}. 
\vspace{0.05in}

{\sc Input}: $A \in {\rm B}_{\rm RC}(n;p)$\\
{\sc Output}: $\bx \triangleq \dec_{{\rm B}_{\rm RC}(n;p)}^{2}(A) \in \{0,1\}^N$, where $N=n^2-n-3$\\[-2mm]

\begin{enumerate}[(I)]
\item Decode the $n$th column, $A^{n}$. If the last bit is 1, flip it to 0 and replace $A^{n;\langle n-1 \rangle}$ with $\phi\Big(A^{n;\langle n-1 \rangle}\Big)$. Otherwise, proceed to the next step. 
\item Decode the $n$th row, $A_{n}$. Check the last bit in $A_{n-1}$, if it is 1, flip it to 0 and replace $A_{n;\langle n-1 \rangle}$ with $\phi\Big(A_{n;\langle n-1 \rangle}\Big)$. Otherwise, proceed to the next step. 
\item Decode the $(n-1)$th row, $A_{n-1}$. Check the second last bit in $A_{n-1}$, if it is 1, flip all the bits in $A_{n-1}$. Otherwise, proceed to the next step. 

Suppose the $n$th row is now $A_{n}=z_1z_2\ldots z_n$.
\item For $1\le i \le (n-1)$, we decode the $i$th column, i.e. $A^{i}$, as follows. If $z_i=1$, replace $A^{i;\langle n-1 \rangle}$ with $\phi\Big(A^{i,\langle n-1 \rangle}\Big)$.

\item Set $\by\triangleq A_{1}  A_{2} \ldots  A_{n-2}  A_{n-1;\langle n-2 \rangle} \in \{0,1\}^{n^2-n-2}$.

\item Output $\bx\triangleq \dec_{\rm seq}(\by) \in \{0,1\}^{n^2-n-3}$.
\end{enumerate}
\vspace{0.05in}

\noindent{\bf Complexity analysis.} For $n\times n$ arrays, it is easy to verify that encoder $\enc_{{\rm B}_{\rm RC}(n;p)}^{2}$ and decoder $\dec_{{\rm B}_{\rm RC}(n;p)}^{2}$ have linear-time complexity. Particularly, there are at most $n+2$ replacements, each replacement is done over sequence of length $n$, and the complexity of encoder/decoder $\enc_{\rm seq}, \dec_{\rm seq}$ is linear over codeword length $N=n^2-n-3=\Theta(n^2)$. We conclude that the running time of encoder $\enc_{{\rm B}_{\rm RC}(n;p)}^{2}$ and decoder $\dec_{{\rm B}_{\rm RC}(n;p)}^{2}$ is $\Theta(n^2)$ which is linear in the message length $N=n^2-n-3$.

\begin{remark}
In \cite{roth:1999}, Talyansky \et{} studied the {\em t-conservative arrays constraint}, where every row has at least $t$ transitions $0 \to 1$ or $1 \to 0$ for some $t\le n/(2\log n)-O(\log n)<n/2$. Such a constraint is equivalent to the $p$-bounded constraint in a weaker condition, where the weight constraint is enforced in every row only. Indeed, for a sequence $\bx=x_1x_2\ldots x_n\in \{0,1\}^n$, consider the {\em differential of $\bx$}, denoted by ${\rm Diff}(\bx)$, which is a sequence $\by=y_1y_2\ldots y_n\in\{0,1\}^n$, where $y_1=x_1$ and $y_i=x_i-x_{i-1} \ppmod{2}$ for $2\le i\le n$. We then observe that $\bx$ has at least $t$ transitions if and only if the weight of ${\rm Diff}(\bx)$ is at least $t$. In addition, the constraint problem where the weight in every row is at least $t$ where $t<n/2$ is equivalent to the constraint problem where the weight in every row is at most $t$ where $t>n/2$. Therefore, one may modify the construction of our proposed encoder $\enc_{{\rm B}_{\rm RC}(n;p)}^{2}$ (i.e for $p=1-1/(2\log n))>1/2$) to construct binary arrays such that there are at least $t$ transitions in every row and every column with at most $n+3$ redundant bits. When the weight-constraint is only required on rows, only Step (I) in $\enc_{{\rm B}_{\rm RC}(n;p)}^{2}$ is sufficient and $N=n^2-1$. Although there is no improvement in the redundancy (the encoder in \cite{roth:1999} also use only one redundant bit), our encoder can be applied for a larger range of $t$, where $t \le n/c$ for any $c>2$. 
\end{remark}

\subsection{Modified SRT to Encode ${\rm Bal}_{\rm RC}(n;\epsilon)$ with One Redundant Bit}

In this section, we show that, for sufficiently large $n$, the redundancy to encode (decode) binary data to (from) ${\rm Bal}_{\rm RC}(n;\epsilon)$ can be further reduced from $\Theta(n)$ bits (from the encoding method in Section III-B) to only a single bit via the SRT.

Similar to the coding method in \cite{TT:2020}, the original data of length $(n^2-1)$ is prepended with 0. We then also remove all the {\em forbidden strings}, and append the replacement strings (starting with {\em marker} 1). To ensure that the encoding process terminates, replacement strings are of shorter lengths compared to the forbidden strings. 

On the other hand, the differences in this work are as follows. There are two types of forbidden strings: those comprising consecutive bits (to ensure the weight constraint over the rows) and those comprising bits that are of distance $n$ bits apart (to ensure the weight constraint over the columns). Consequently, we have two types of markers. 
\begin{definition}\label{forbidden-definition} 
For a binary sequence $\bx=x_1x_2\ldots x_m$, a subsequence $\by$ of $\bx$ is said to be $(\ell,\epsilon)$-{\em r-forbidden} if $\by=x_{i}x_{i+1}\ldots x_{i+\ell-1}$ for some $i$ and $\by$ is not $\epsilon$-balanced. On the other hand, a subsequence $\bz$ of $\bx$ is said to be  $(\ell,\epsilon)$-{\em c-forbidden} if $\bz=x_{j}x_{j+n}\ldots x_{j+(\ell-1)n}$ for some $j$ and $\bz$ is not $\epsilon$-balanced.
\end{definition}

Given $\epsilon>0$, let ${\bf F}(\ell,\epsilon)$ denote the set of all forbidden sequences of size $\ell$, that are not $\epsilon$-balanced. The following theorem provides an upper bound on the size of ${\bf F}(\ell,\epsilon)$.

\begin{theorem}\cite[Theorem 5]{TT:2020}\label{thm2} For $\epsilon>0, m\ge16$ and $\ell\le m$ such that $(1/\epsilon^2) \ln m \le \ell$, let 
$k=\ell-3-\log m$, there exists an one-to-one map $\Psi: {\bf F}(\ell,\epsilon) \to \{0,1\}^k$.
\end{theorem} 

The look-up table for $\Psi: {\bf F}(\ell,\epsilon) \to \{0,1\}^k$, which is roughly of size $2^\ell$, is needed for our encoding and decoding algorithms. To obtain efficient algorithms whose running time is linear in $m$, we set $\ell=\alpha \ln m=\Theta(\log m)$, for some fixed $\alpha \ge 1/\epsilon^2$. 
  
Note that, when all forbidden strings have been removed, the length of the current encoded sequence is strictly smaller than $n^2$. Similar to the coding methods in \cite{srt:2019,TT:2020}, we introduce the {\em extension phase} procedure to append bits to obtain an encoded sequence of length $n^2$ while the weight constraint is still preserved. Crucial to the correctness of our encoding algorithm is the following lemma. 
\begin{lemma}\label{epsilon2}
Given $\epsilon>0$, integers $n,\ell,s$, where $n=\ell s$, $n\epsilon\ge 2$. Set $\epsilon'=\epsilon/2$ and suppose that $\bx \in \{0,1\}^{n-1}$ where every window of size $\ell$ of $\bx$ is $\epsilon'$-balanced. We then have $\bx0$ and $\bx1$ that are both $\epsilon$-balanced. 
\end{lemma}
\begin{proof}
Since every window of size $\ell$ of $\bx$ is $\epsilon'$-balanced, we then have 
\begin{align*}
{\rm wt}(\bx1) \ge {\rm wt}(\bx0) &\ge (1/2-\epsilon')\ell s-1 \\
&=(1/2-\epsilon)n+n\epsilon/2- 1 \\
&\ge (1/2-\epsilon)n. 
\end{align*}
Similarly, we have
\begin{align*}
{\rm wt}(\bx0) \le {\rm wt}(\bx1) &\le  (1/2+\epsilon')\ell s+1 \\
&=(1/2+\epsilon)\ell s +(1-\epsilon\ell s/2) \\
&=(1/2+\epsilon)n +(1-\epsilon n/2) \\
&\le (1/2+\epsilon)n. 
\end{align*}
Thus, $\bx0$ or $\bx1$ are both $\epsilon$-balanced.
\end{proof}

We now present a linear-time algorithm $\enc_{{\rm Bal}_{\rm RC}(n;\epsilon)}^{2}$ to encode ${\rm Bal}_{\rm RC}(n;\epsilon)$ that incurs at most one redundant bit. 
\vspace{0.05in}

\noindent {\bf Encoder} $\enc_{{\rm Bal}_{\rm RC}(n;\epsilon)}^{2}$. Given $n, \epsilon>0$, we set $m=n^2, \epsilon'=\epsilon/2, \ell=\ceil{\alpha \log_e m}$ for some fixed number $\alpha\ge 1/\epsilon'^2$.  The source sequence $\bx\in\{0,1\}^{n^2-1}$. For simplicity, we assume $\log n$ is an integer, and $n=\ell s$ for some integer $s$. According to Theorem~\ref{thm2}, there exists an one-to-one map $\Psi: {\bf F}(\ell,\epsilon') \to \{0,1\}^k$, where $k=\ell-3-\log m=\ell-3-2\log n$. Here, $ {\bf F}(\ell,\epsilon')$ denotes the set of all forbidden sequences of size $\ell$, that are not $\epsilon'$-balanced and can be mapped one-to-one to binary sequences of length $\ell-3-2\log n$. 

 The algorithm contains three phases: {\em initial phase}, {\em replacement phase} and {\em extension phase}. Particularly, the extension phase includes {\em row extension} and {\em array extension}. 
\vspace{0.05in}

\noindent{\bf Initial phase.} The source sequence $\bx\in\{0,1\}^{n^2-1}$ is prepended with $0$, to obtain $\bc=0\bx \in \{0,1\}^{n^2}$. The encoder scans $\bc$ and if there is no forbidden subsequence, neither $(\ell,\epsilon')$-r-forbidden nor $(\ell,\epsilon')$-c-forbidden, it outputs $\Phi^{-1}(\bc)$, which is an array of size $n\times n$. Otherwise, it proceeds to the replacement phase. Observe that if there is no $(\ell,\epsilon')$-r-forbidden or $(\ell,\epsilon')$-c-forbidden in all rows and all columns then all rows and all columns are $\epsilon'$-balanced, and since $\epsilon'=\epsilon/2<\epsilon$, all rows and all columns are then $\epsilon$-balanced.
\vspace{0.05in}

\noindent{\bf Replacement phase.} Let the current word $\bc$ of length $N_0$. In the beginning, $N_0=n^2$. The encoder searches for the forbidden subsequences. Suppose the first forbidden subsequence starts at $x_i$ for some $1\le i\le N_0-\ell+1$. Let $\bp$ be the binary representation of the index $i$ of length $2\log n$. Let $\by=x_{i}x_{i+1}\ldots x_{i+\ell-1}$ and $\bz=x_{i}x_{i+n}\ldots x_{i+(\ell-1)n}$. 
\begin{itemize}
\item If $\by$ is $(\ell,\epsilon')$-r-forbidden, the encoder sets ${\rm R}=11 \bp \Psi(\by)$. It then removes $\by$ from $\bc$ and prepends ${\rm R}$ to $\bc$. The encoder repeats the replacement phase. 
\item  If $\by$ is not $(\ell,\epsilon')$-r-forbidden, and $\bz$ is $(\ell,\epsilon')$-c-forbidden, the encoder sets ${\rm R}=10 \bp \Psi(\bz)$. It then removes $\bz$ from $\bc$ and prepends ${\rm R}$ to $\bc$. The encoder repeats the replacement phase. 
\end{itemize} 
The encoder exits the replacement phase and proceeds to the extension phase if, after some replacement, the current sequence $\bc$ contains no $(\ell,\epsilon')$-r-forbidden (or $(\ell,\epsilon')$-c-forbidden) subsequence, or the current sequence is of length $n/2$. Otherwise, the encoder repeats the replacement phase. Note that such a replacement operation reduces the length of the sequence by one, since we remove a subsequence of $\ell$ bits and replace it by $2+2\log n+k=2+2\log n+(\ell-3-2\log n)=\ell-1$ (bits). Therefore, this procedure is guaranteed to terminate. 
We illustrate the idea of the replacement phase through Figure 2.
\vspace{0.05in}

\noindent{\bf Extension phase.} If the length of the current sequence $\bc$ is $N_0$ where $N_0<n$, the encoder appends a suffix of length $N_1=n^2-N_0$ to obtain a sequence of length $n^2$. Note that at the end of the replacement phase, the length of the current sequence is at least $n/2$. Suppose that $N_0=n \times q+r$ where $0\le q< n$, $0\le r<n$. If we fill the bits in $\bc$ to an array of size $n\times n$, we can fill $q$ rows and the $(q+1)$th row includes $r$ bits. From the replacement phase, in the worst case scenario, we have $N_0=n/2, q=0, r=n/2$. The extension phase includes two steps: the row extension step \big(to fulfill the $(q+1)$th row\big) and the array extension step \big(to fulfill the remaining $(n-q-1)$ rows\big). 

\noindent{\bf Row extension.} The encoder fulfills the $(q+1)$th row as follows. 
\begin{itemize} 
\item If $N_0=n/2$, i.e. $q=0$, $r=n/2$, the encoder simply concatenate the sequence $\bc$ and its complement $\overline{\bc}$ to obtain the first row. In this case, it is easy to see that the row is balanced. 
\item If $N_0>n/2>\ell$, we then observe that the output sequence from the replacement phase, $\bc$, contains no $(\ell,\epsilon')$-r-forbidden subsequence of length $\ell$. Let $\bw$ be the suffix of length $\ell$ of $\bc$. The encoder simply appends $\bw$ repeatedly until the $(q+1)$th row is fulfilled. Formally, let $j$ be the smallest integer such that $\bc' = \bc\bw^j$ is of length greater than $(q+1)n$. The encoder outputs the prefix of length $(q+1)n$ of $\bc'$. In this case, it is also easy to verify that, after the row extension phase, the $(q+1)$th row is $\epsilon'$-balanced. Since $\bc$ does not contain any $(\ell,\epsilon')$-r-forbidden subsequence, it remains to show that there is no $(\ell,\epsilon')$-r-forbidden subsequence in the suffix $\bw^j$. It is easy to see that repeating the vector $\bw$ clearly satisfies the constraint since every subsequence of size $\ell$ generated in this manner is a cyclic shift of the vector $\bw$, and since $\bw$ is $\epsilon'$-balanced, there is no $(\ell,\epsilon')$-r-forbidden subsequence.
\end{itemize}



\noindent{\bf Array extension.} After the row extension step, the array is of size $(q+1) \times n$, where $0\le q< n$, and every row is $\epsilon'$-balanced.
\begin{itemize}
\item If $q+1=n$ or $q=n-1$, the encoder simply outputs this array. Since every row is $\epsilon'$-balanced, which is also $\epsilon$-balanced, it remains to show that every column is $\epsilon$-balanced. From the replacement step, every window of size $\ell$ of the first $(n-1)$ bits in every column is $\epsilon'$-balanced. According to Lemma~\ref{epsilon2}, we then have every column is $\epsilon$-balanced.  
\item If $q+1 \le n/2$, the encoder simply fills in the next $(q+1)$ rows with the complement of the current array. After that the encoder fills the remaining rows alternately with balanced sequences $\bw=0101\ldots=(01)^{n/2}$, and $\bw'=1010\ldots=(10)^{n/2}$. In this case, it is easy to show that every row and every column is $\epsilon$-balanced. 
\item If $q+1 > n/2$, similar to the process in the row extension step, for each column $i$, the encoder sets $\by_i$ to be the sequence obtained by the first $q$ bits (i.e. except the bit in the $(q+1)$th row) and $\bw_i$ to be the suffix of length $\ell$ of $\by_i$. It then appends $\bw_i$ repeatedly until the $i$th column is fulfilled. Similar to the proof in the row extensive phase, for every column, the sequence obtained by $(n-1)$ bits, except the bit in the $(q+1)$th row, is $\epsilon'$-balanced. Again, according to Lemma~\ref{epsilon2}, every column is $\epsilon$-balanced. Note that all $(n-q-1)$ rows, that have been fulfilled, are actually repetition of some previous rows (i.e. the $j$th column is filled exactly as the $j-\ell-1$th column), therefore, they are all $\epsilon'$-balanced, and hence, are also $\epsilon$-balanced.
\end{itemize}
\vspace{0.05in}

The following result is then immediate.
\begin{theorem}
The Encoder $\enc_{{\rm Bal}_{\rm RC}(n;\epsilon)}^{2}$ is correct, i.e. $\enc_{{\rm Bal}_{\rm RC}(n;\epsilon)}^{2}(\bx) \in {\rm Bal}_{\rm RC}(n; \epsilon)$ for all $\bx\in \{0,1\}^{n^2-1}$. 
\end{theorem}

For completeness, we describe the corresponding decoder $\dec_{{\rm Bal}_{\rm RC}(n;\epsilon)}^{2}$ as follows. 
\vspace{0.05in}

\noindent {\bf Decoder} $\dec_{{\rm Bal}_{\rm RC}(n;\epsilon)}^{2}$. 

{\sc Input}: $A \in {\rm Bal}_{\rm RC}(n;\epsilon)$\\
{\sc Output}: $\bx \triangleq \dec_{{\rm Bal}_{\rm RC}(n;\epsilon)}^{2}(A) \in \{0,1\}^{n^2-1}$\\[-2mm]

\noindent {\bf Decoding procedure.} From an array A of size $n\times n$, the decoder first obtains the binary sequence $\bx=\Phi^{-1}(A)$ of length $n^2$. The decoder scans from left to right. If the first bit is 0, the decoder simply removes 0 and identifies the last $(n^2-1)$ bits are source data. On the other hand, if it starts with 11, the decoder takes the prefix of length $(\ell-1)$ and concludes that this prefix is obtained by a replacement of $(\ell,\epsilon')$-r-forbidden subsequence. In other words, the prefix is of the form $11\bp\Psi(\by)$, where $\bp$ is of length $2\log n$ and $\Psi(\by)$ is of length $k$. The decoder removes this prefix, adds the subsequence $\by=\Psi^{-1}(\Psi(\by))$ into position $i$, which takes $\bp$ as the binary representation. On the other hand, if it starts with 10, the decoder takes the prefix of length $(\ell-1)$ and concludes that this prefix is obtained by a replacement of $(\ell,\epsilon')$-c-forbidden subsequence. In other words, the prefix is of the form $10\bp\Psi(\bz)$, where $\bp$ is of length $2\log n$ and $\Psi(\bz)$ is of length $k$. The decoder removes this prefix, and adds the forbidden subsequence $\bz=\Psi^{-1}(\Psi(\bz))$ into position $i$, which takes $\bp$ as the binary representation. It terminates when the first bit is 0, and simply takes the following $(n^2-1)$ bits as the source data.

We illustrate the idea of the extension phase through the following example. 
\begin{example} Consider $n=6, \ell=3, \epsilon=1/6$, i.e. the weight of every row and every column is within $[2,4]$. The first example considers the worst case scenario.

\begin{small}
 \[\left( \begin{array}{cccccc}
1 & 1 & 0 & ? & ? & ?\\
? & ? & ? & ? & ? & ?\\
? & ? & ? & ? & ? & ?\\
? & ? & ? & ? & ? & ?\\
? & ? & ? & ? & ? & ?\\
? & ? & ? & ? & ? & ?\end{array} \right) \to \left( \begin{array}{cccccc}
1 & 1 & 0 & {\color{blue}{0}} & {\color{blue}{0}} & {\color{blue}{1}}\\
{\color{green}{0}} & {\color{green}{0}} & {\color{green}{1}} & {\color{green}{1}} & {\color{green}{1}} & {\color{green}{0}}\\
0 & 1 & 0 & 1 & 0 & 1\\
1 & 0 & 1 & 0 & 1 & 0\\
0 & 1 & 0 & 1 & 0 & 1\\
1 & 0 & 1 & 0 & 1 & 0\end{array} \right)\]
\end{small}
In this example, $N_0=n/2=3$ and $q=0$. The bits in blue denote the complement of the first $n/2$ bits in the row extension step, the bits in green denote the complement of the first row in the array extension step since $q+1=1<n/2$, and the last three rows are filled by sequences $\bw=010101$ and $\bw'=101010$ alternately. 

Another example is as follows.

\begin{small}
 \[\left( \begin{array}{cccccc}
1 & 1 & 0 & 1 & 0 & 1\\
0 & 1 & 0 & 1 & 0 & 0\\
1 & 0 & 1 & 0 & 1 & 0\\
0 & 1 & 1 & 0 & 0 & 1\\
0 & ? & ? & ? & ? & ?\\
? & ? & ? & ? & ? & ?\end{array} \right) \to \left( \begin{array}{cccccc}
1 & 1 & 0 & 1 & 0 & 1\\
0 & 1 & 0 & 1 & 0 & 0\\
1 & 0 & 1 & 0 & 1 & 0\\
0 & 1 & 1 & 0 & {\color{blue}{0}} & {\color{blue}{1}}\\
{\color{blue}{0}} & {\color{red}{0}} & {\color{red}{1}} & {\color{red}{0}} & {\color{red}{0}} & {\color{red}{1}}\\
{\color{green}{0}} & {\color{green}{1}} & {\color{green}{0}} & {\color{green}{1}} & {\color{green}{0}} & {\color{green}{0}}\end{array} \right)\]
\end{small}
In this example, $N_0=25$ and $q=4$. The encoder first fulfills the $(q+1)$th or the fifth row, and then proceeds to the array extension. 
Here, the bits in blue denote the suffix of length $\ell$ bits before the row extension step, and the bits in red denote the process that the encoder repeats the suffix until the fifth row is fulfilled, and finally the last row is filled as the repetition of the second row. Throughout the two examples, we can verify that every row and every column is $\epsilon$-balanced.
\end{example}

\begin{remark}\label{bound-for-n}
We now discuss the lower bound for $n$ so that the encoding algorithm works. Given $\epsilon\in (0,1/2)$, we require $n^2 \ge 1/\epsilon'^2 \ln n^2=8/\epsilon^2 \ln n$ and $n\epsilon \ge 2$. Since $n\gg \ln n$, a simple lower bound would be $n\ge 8/\epsilon^2$. Furthermore, when $p>1/2$, if we set $\epsilon=p-1/2$ then ${\rm Bal}_{\rm RC}(n;\epsilon) \subset {\rm B}_{\rm RC}(n;p)$. Hence, one may use our encoding method for ${\rm Bal}_{\rm RC}(n;\epsilon)$ to encode ${\rm B}_{\rm RC}(n;p)$, which uses one redundant bit. This improves the redundancy of the encoder in Subsection IV-A, which incurs $n+3$ redundant bits. 
\end{remark}



\begin{figure}[!t]
	Suppose that the forbidden subsequence that begins at index $i$. We set $\bp$ be the binary representation of $i$ of length $2\log n$. 
	\begin{itemize}
	
	\item [(a)] When $\by$ is an $(\ell,\epsilon')$-r-forbidden subsequence. 
	
	\vspace{1mm}
	
	\begin{center}
	\begin{tikzpicture}[x=1cm,y=1cm]
		
		\arraycolsep=3pt
		\tikzset{ff/.style = {rectangle, draw=red, line width=2pt, align=center}}
		\tikzset{cc/.style = {rectangle, draw=black, line width=2pt, align=center}}
		\tikzset{rr/.style = {rectangle, draw=blue, line width=2pt, align=center}}
		\tikzset{arrow/.style = {->,> = {Stealth[scale=1.5]},black,line width=2pt}}
		
		\node[cc, minimum height=6mm,minimum width=2cm] at (-2,0) {$\bc_1$};
		\node[cc, minimum height=6mm,minimum width=3cm] at (2.5,0) {$\bc_2$};
		\node[ff, minimum height=6mm,minimum width=2cm] at (0,0) {$\by$};
  		
		
		\Dimline[(-1,0.4)][(1,0.4)][$\ell$] ;
		\Dimline[(-1,-2)][(-3,-2)][$\ell-1$] ;
		\node[cc, minimum height=6mm,minimum width=2cm] at (0,-1.5) {$\bc_1$};
		\node[cc, minimum height=6mm,minimum width=3cm] at (2.5,-1.5) {$\bc_2$};
		\node[rr, minimum height=6mm,minimum width=2cm] at (-2,-1.5) {$11\bp\Psi(\by)$};
		
		\draw[arrow] (0,-0.3) to (-1.2,-1.2);
	\end{tikzpicture}
	\end{center}

	\item [(b)] When $\by$ is not $(\ell,\epsilon')$-r-forbidden subsequence and so, $\bz$ is $(\ell,\epsilon')$-c-forbidden. 
	
	\vspace{1mm}
	
	\begin{center}
		\begin{tikzpicture}[x=1cm,y=1cm]
			
			\arraycolsep=3pt
			\tikzset{ff/.style = {rectangle, draw=red, fill=white, line width=2pt, align=center, text=red}}
			\tikzset{cc/.style = {rectangle, draw=black, line width=2pt, align=center}}
			\tikzset{rr/.style = {rectangle, draw=blue, line width=2pt, align=center, text=blue}}
			\tikzset{arrow/.style = {->,> = {Stealth[scale=1.5]},black,line width=2pt}}
			
			\node[cc, minimum height=6mm,minimum width=1.5cm] at (-2.75,0) {$\bc_1$};
			\node[cc, minimum height=6mm,minimum width=1cm] at (-1,0) {$\bc_2$};
			\node[cc, minimum height=6mm,minimum width=1cm] at (0.5,0) {$\bc_3$};
			\node[cc, minimum height=6mm,minimum width=1cm] at (2,0) {$\cdots$};
			\node[cc, minimum height=6mm,minimum width=1cm] at (3.5,0) {$\bc_{\ell+1}$};
			
			\node[ff, minimum height=6mm,minimum width=5mm](z1) at (-1.75,0) {$z_1$};
			\node[ff, minimum height=6mm,minimum width=5mm](z2) at (-0.25,0) {$z_2$};
			\node[ff, minimum height=6mm,minimum width=5mm](z3) at (1.25,0) {$z_3$};
			\node[ff, minimum height=6mm,minimum width=5mm](zn) at (2.75,0) {$z_{\ell}$};
			
			\node[ff, minimum height=6mm,minimum width=2cm](z) at (0.5,-1.5) {$\bz$};
			
			\draw[-{latex}] (z1.south) to (z.north);
			\draw[-{latex}] (z2.south) to (z.north);
			\draw[-{latex}] (z3.south) to (z.north);
			\draw[-{latex}] (zn.south) to (z.north);
			
			\Dimline[(1.5,-2)][(-0.5,-2)][$\ell$];
			\Dimline[(-1.75,-3.5)][(-3.75,-3.5)][$\ell-1$] ;
			\Dimline[(-2,0.5)][(-0.5,0.5)][$n$];
			\Dimline[(-0.5,0.5)][(1,0.5)][$n$];
			
			\node[cc, minimum height=6mm,minimum width=2cm] at (-0.75,-3) {$\bc_1$};
			\node[cc, minimum height=6mm,minimum width=1cm] at (0.75,-3) {$\bc_2$};
			\node[cc, minimum height=6mm,minimum width=1cm] at (1.75,-3) {$\bc_3$};
			\node[cc, minimum height=6mm,minimum width=1cm] at (2.75,-3) {$\cdots$};
			\node[cc, minimum height=6mm,minimum width=1cm, fill=white] at (3.75,-3) {$\bc_{\ell+1}$};
			
			\node[rr, minimum height=6mm,minimum width=2cm](R) at (-2.75,-3) {$10\bp\Psi(\bz)$};
			
			\draw[arrow] (z.west) to (R.north);
		\end{tikzpicture}
	\end{center}

	\end{itemize}
	
	\caption{$(\ell,\epsilon')$-r-forbidden and $(\ell,\epsilon')$-c-forbidden replacement.
	}
	\label{fig2}
\end{figure}
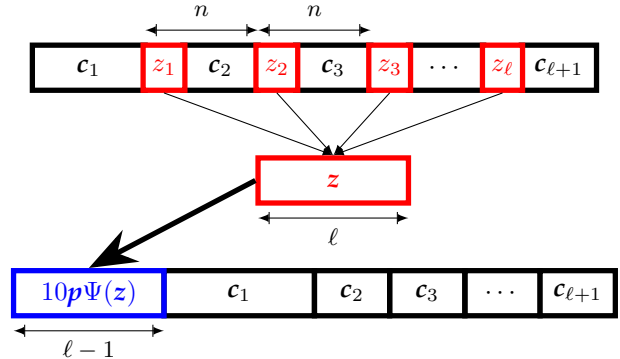





\section{2D Subarray Constrained Codes}

In this section, we are interested in the problem of designing efficient coding methods that encode (decode) binary data to (from) ${\rm B}_{\rm S}(n,m;p)$ and ${\rm Bal}_{\rm S}(n,m;\epsilon)$. Similar to the case of constructing 2D row/column constrained codes ${\rm B}_{\rm RC}(n;p)$ and ${\rm Bal}_{\rm RC}(n;\epsilon)$, the challenge in coding design is to find an efficient method to enforce the weight constraint in every subarray so that changing the weight of a subarray does not violate the weight-constraint in previous coded subarrays. 

Our main results in this section are summarised as follows.
\begin{itemize}
\item We use the construction of antipodal matching (as discussed in Section II-B and Section IV-A) to encode (decode) binary data to (from) ${\rm B}_{\rm S}(n,m;p)$ with at most $n$ redundant bits when $m=n-\Theta(1)$ and $p\ge 1/2$. The construction can be extended to obtain capacity-approaching encoder with at most $o(n^2)$ redundant bits when $m=n-o(n)$. 
\item We use SRT to encode (decode) binary data to (from) ${\rm Bal}_{\rm S}(n,m;\epsilon)$. For sufficiently large $n$, this method incurs at most one redundant bit. 
\end{itemize} 

We first recall the sets ${\rm B}_{\rm S}(n,m;p)$ and ${\rm Bal}_{\rm S}(n,m;\epsilon)$. Given $n,m,p,\epsilon$, where $m\le n, \epsilon\in [0,1/2], p\in[0,1]$ we set 
\begin{align*}
{\rm B}_{\rm S}(n,m;p) &\triangleq \Big\{A \in {\bA}_n: \text{every subarray } B \text{ of size } m\times m 
\text{ are $p$-bounded} \Big\}, \text{ and } \\
{\rm Bal}_{\rm S}(n,m;\epsilon)  &\triangleq \Big\{A \in {\bA}_n: \text{every subarray } B \text{ of size } m\times m  
\text{ are $\epsilon$-balanced} \Big\}.
\end{align*} 

\begin{example}
Consider $n=4, m=2, p=1/4$ and two given arrays $A$ and $B$ as follows. We observe that $A\in {\rm B}_{\rm RC}(n;p)$, however $A\notin {\rm B}_{\rm S}(n,m;p)$ since there is a subarray of size $2\times 2$ (highlighted in red), which is not $p$-bounded. On the other hand, we have $B\in {\rm B}_{\rm S}(n,m;p)$, and $B\notin {\rm B}_{\rm RC}(n;p)$ as the first column of $B$ (highlighted in blue) is not $p$-bounded.
 \[A=\left( \begin{array}{cccc}
0 & 0 & 0 & 0\\
0 & 0 & 0 & 0\\
0 & 0 & {\color{red}{0}} & {\color{red}{1}} \\
0 & 0 & {\color{red}{1}}  & {\color{red}{0}} \end{array} \right), B= \left( \begin{array}{cccc}
{\color{blue}{1}}  & 0 & 0 & 0\\
{\color{blue}{0}} & 0 & 1 & 0\\
{\color{blue}{1}} & 0 & 0 & 0\\
{\color{blue}{0}} & 0 & 0 & 1\end{array} \right)\]
\end{example}



\subsection{Antipodal Matching for ${\rm B}_{\rm S}(n,m;p)$}
Suppose that $m=n-k$ for some constant number $k=\Theta(1)$. Observe that ${\rm B}_{\rm S}(n,m;1/2) \subset {\rm B}_{\rm S}(n,m;p)$ for all $p\ge 1/2$. Hence, we aim to encode (decode) binary data to (from) ${\rm B}_{\rm S}(n,m;1/2)$ with at most $n$ redundant bits, i.e. for any array $A\in {\rm B}_{\rm S}(n,m;1/2)$, the weight of every subarray of size $m\times m$ is at most $m^2/2$. For simplicity, we suppose that $m$ is even and $2(k+1)^2 \le n$.
\vspace{0.05in}

We follow \cite{ROTH:ISIT} to construct the antipodal matchings $\phi_1$ for sequences of length $m^2$ and $\phi_2$ for sequences of length $m^2-m$. In other words, $\phi_1: \{0,1\}^{m^2} \to  \{0,1\}^{m^2}$ and $\phi_2: \{0,1\}^{m^2-m} \to  \{0,1\}^{m^2-m}$ such that for arbitrary $\bx \in \{0,1\}^{m^2}$ and $\by \in \{0,1\}^{m^2-m}$:
\begin{itemize}
\item ${\rm wt}(\phi_1(\bx)) = n - {\rm wt}(\bx),$ and ${\rm wt}(\phi_2(\by)) = n - {\rm wt}(\by),$
\item If ${\rm wt}(\bx)>m^2/2$  then $\phi_1(\bx)$  has all its 1's in positions where $\bx$ has 1's. In other words, suppose $\bx=x_1x_2\ldots x_{m^2}$ and $\bz=\phi_1(\bx)=z_1z_2\ldots z_{m^2}$, then $z_i =1$ implies $x_i =1$ for  $1\le i\le m^2$. We have the same argument for $\by$.
\item $\phi_1(\phi_1(\bx)) = \bx$ and $\phi_2(\phi_2(\by)) = \by$.
\end{itemize}
\vspace{0.05in}

We now describe the detailed construction of the subarray $p$-bounded encoder,  $\enc_{{\rm B}_{\rm S}(n,m;p)}$ when $p\ge 1/2$. 

\vspace{0.05in}
\noindent{\bf Subarray $p$-bounded encoder, $\enc_{{\rm B}_{\rm S}(n,m;p)}$}. 
\vspace{0.05in}


{\sc Input}: $\bx \in \{0,1\}^{n^2-n}$\\
{\sc Output}: $A \triangleq \enc_{{\rm B}_{\rm S}(n,m;p)}(\bx) \in {\rm B}_{\rm S}(n,m;1/2) \subset {\rm B}_{\rm S}(n,m;p)$\\[-2mm]

\begin{enumerate}[(I)]
\item Fill $n^2-n$ bits of $\bx$ to $A$ row by row to obtain a subarray of size $(n-1)\times n$ and suppose the last row of $A$ is $A_n=y_1y_2\ldots y_n$. 

\item Observe that there are $(k+1)^2$ subarrays of size $m\times m$ that we need to enforce the weight constraint. We set the order of the subarray as row by row and let $B_i$ be the $i$th subarray of $A$ for $1\le i\le (k+1)^2$. 

\item Using $\phi_1$ for $B_i$ where $1\le i\le k(k+1)$:
\begin{itemize}
\item If ${\rm wt}(B_i)>m^2/2$, set $y_{2i-1}=1$ and replace the entire subarray $B_i$ with $\phi_1(B_i)$.
\end{itemize}

\item Using $\phi_2$ for $(k+1)$ subarrays of size $(m-1)\times m$:
\begin{itemize}
\item For $k(k+1)+1\le i\le (k+1)^2$, set $C_i$ be the subarray obtained by removing the last row of $B_i$. In other words, $C_i$ is of size $(m-1)\times m$ which does not include the bits of the last row $A_n$. 
\item If ${\rm wt}(C_i)>(m^2-m)/2$, set $y_{2i-1}=1$ and replace the entire subarray $C_i$ with $\phi_2(C_i)$.
\end{itemize}

\item Filling the remaining bit of the $n$th row: 
\begin{itemize}
\item Set $y_{2i}=\overline{y_{2i-1}}$ for $1\le i\le (k+1)^2$. We then obtain $y_1y_2\ldots y_{2(k+1)^2}$ as a balanced sequence. 
\item Fill the remaining bits of the $n$th row with 0.  
\end{itemize}

\item Output $A$.
\end{enumerate}

\begin{theorem}\label{constrained proof}
The Encoder  $\enc_{{\rm B}_{\rm S}(n,m;p)}$ is correct. In other words,  $\enc_{{\rm B}_{\rm S}(n,m;p)}(\bx) \in {\rm B}_{\rm S}(n,m;1/2)$ for all $\bx \in \{0,1\}^{n^2-n}$. The redundancy is $n$ (bits).
\end{theorem}

\begin{proof}
Suppose that $A=\enc_{{\rm B}_{\rm S}(n,m;p)}(\bx)$ for some $\bx  \in \{0,1\}^{n^2-n}$. We now show that every subarray of size $m\times m$ have weight at most $m^2/2$. We set the order of the subarray as row by row and let $B_i$ be the $i$th subarray of $A$ for $1\le i\le (k+1)^2$. 

We first show that the weight of $B_i$ is at most $m^2/2$ for $1\le i\le k(k+1)$.
From step (III), suppose that a subarray $B_i$ satisfy the weight constraint and the encoder proceeds to replace some subarray $B_j$ with $\phi_1(B_j)$ where $B_j$ has some overlapping bits with $B_i$ and ${\rm wt}(B_j) > m^2/2$. Note that $\phi_1(B_j)$ has all its 1's in positions where $B_j$ has 1's and ${\rm wt}(\phi_1(B_j))\le m^2/2$. Therefore, whenever the encoder performs replacement in $B_j$, it does not increase the weight of $B_i$, the '0' bits in $B_i$ will not change to '1', only the '1' bits in $B_i$ will either stay as 1 or change to 0. Hence, it does not violate the weight constraint in $B_i$. Thus, at the end of step (III), the encoder ensures that the weight of $B_i$ is at most $m^2/2$ for $1\le i\le k(k+1)$. Similarly, at the end of step (IV), we have the weight of $C_i$ is at most $(m^2-m)/2$ for $k(k+1)+1\le i\le (k+1)^2$. 

It remains to show that the weight of $B_i$ is at most $m^2/2$ for $k(k+1)+1\le i\le (k+1)^2$. Observe that, from the step (V), every $m$ consecutive bits of the last row $A_n$ form a balanced sequence (here $m$ is even). For $k(k+1)+1\le i\le (k+1)^2$, each $B_i$ is the concatenation of $C_i$ (the weight of $C_i$ is at most $(m^2-m)/2$) and $m$ consecutive bits of the last row $A_n$ (which form a balanced sequence). Thus, the weight of $B_i$ is also at most $(m^2-m)/2+m/2=m^2/2$. 
\end{proof}
For completeness, we describe the corresponding decoder $\dec_{{\rm B}_{\rm S}(n,m;p)}$ as follows. 

\vspace{0.05in}
\noindent{\bf Subarray $p$-bounded decoder, $\dec_{{\rm B}_{\rm S}(n,m;p)}$}. 
\vspace{0.05in}

{\sc Input}: $A \in {\rm B}_{\rm S}(n,m;p)$\\
{\sc Output}: $\bx \triangleq\dec_{{\rm B}_{\rm S}(n,m;p)}(A) \in \{0,1\}^{n^2-n}$\\[-2mm]

\begin{enumerate}[(I)]
\item Suppose that $A_n=y_1y_2\ldots y_n$. We set the order of the subarray as row by row and let $B_i$ be the $i$th subarray of $A$ of size $m\times m$ for $1\le i\le (k+1)^2$. For $k(k+1)+1\le i\le (k+1)^2$, set $C_i$ be the subarray obtained by removing the last row of $B_i$. In other words, $C_i$ is of size $(m-1)\times m$. 
\item For $1\le i\le k(k+1)$, if $y_{2i-1}=1$, replace $B_i$ with $\phi_1(B_i)$.  
\item For $k(k+1)+1\le i\le (k+1)^2$, if $y_{2i-1}=1$, replace $C_i$ with $\phi_2(C_i)$. 
\item Set $A'$ be the array obtained by the first $(n-1)$ rows of the current array.
\item Output $\bx\triangleq \Phi^{-1}(A') \in \{0,1\}^{n^2-n}$.
\end{enumerate}
\vspace{0.05in}

We illustrate the idea of the encoding process through the following example.

\begin{example} 
Consider $n=9, m=8$, i.e. the weight of every subarray of size $8\times 8$ is at most 32. We construct two antipodal matching: $\phi_1$ for sequences of length 64 and $\phi_2$ for sequences of length 56. 
There are four subarrays of $A$ to enforce the weight constraint and the last row $A_9=y_1y_2\ldots y_9$ is for the redundant bits. 
\begin{small}
 \[A=\left( \begin{array}{ccccccccc}
1 & 0 & 0 & 0 & 0 & 0& 1 & 0 & 1\\
0 & 1 & 1 & 1 & 1 & 0& 0 & 0 & 1\\
0 & 0 & 1 & 1 & 1 & 1& 1 & 1 & 1\\
0 & 0 & 0 & 1 & 1 & 0& 0 & 1 & 1\\
1 & 1 & 1 & 1 & 1 & 1& 1 & 1 & 1\\
1 & 0 & 1 & 0 & 0 & 1& 1 & 0 & 1\\
1 & 1 & 1 & 0 & 0 & 1& 1 & 1 & 1\\
0 & 0 & 1 & 0 & 1 & 1& 1 & 1 & 1\\
y_1 & y_2 & y_3 & y_4 & y_5 & y_6& y_7 & y_8 & y_9\end{array} \right) \]
\end{small}
The encoder first checks $B_1$ and since ${\rm wt}(B_1)=38>32$, it replaces $B_1$ with $\phi_1(B_1)$ (as highlighted in blue), where ${\rm wt}(\phi(B_1))=26$, and sets $y_1=1$. 
\begin{small}
 \[A=\left( \begin{array}{ccccccccc}
{\color{blue}{0}} & {\color{blue}{0}} & {\color{blue}{0}} & {\color{blue}{0}} & {\color{blue}{0}} & {\color{blue}{0}}& {\color{blue}{1}} &{\color{blue}{0}} & 1\\
{\color{blue}{0}} & {\color{blue}{1}} & {\color{blue}{1}} & {\color{blue}{0}} & {\color{blue}{0}} & {\color{blue}{0}}& {\color{blue}{0}} & {\color{blue}{0}} & 1\\
{\color{blue}{0}} & {\color{blue}{0}} & {\color{blue}{1}} & {\color{blue}{0}} & {\color{blue}{1}} & {\color{blue}{1}} & {\color{blue}{1}} & {\color{blue}{1}} & 1\\
{\color{blue}{0}} & {\color{blue}{0}} & {\color{blue}{0}} & {\color{blue}{1}} & {\color{blue}{1}} & {\color{blue}{0}} & {\color{blue}{0}} & {\color{blue}{0}} & 1\\
{\color{blue}{0}} & {\color{blue}{1}} & {\color{blue}{1}} & {\color{blue}{1}} & {\color{blue}{1}} & {\color{blue}{1}} & {\color{blue}{0}} & {\color{blue}{1}} & 1\\
{\color{blue}{0}} & {\color{blue}{0}} & {\color{blue}{1}} & {\color{blue}{0}} & {\color{blue}{0}} & {\color{blue}{1}} & {\color{blue}{1}} & {\color{blue}{0}} & 1\\
{\color{blue}{0}} & {\color{blue}{1}} & {\color{blue}{1}} & {\color{blue}{0}} & {\color{blue}{0}} & {\color{blue}{0}} & {\color{blue}{0}} & {\color{blue}{0}} & 1\\
{\color{blue}{0}} & {\color{blue}{0}} & {\color{blue}{1}} & {\color{blue}{0}} & {\color{blue}{1}} & {\color{blue}{1}} & {\color{blue}{1}} & {\color{blue}{1}} & 1\\
{\color{red}{1}} & y_2 & y_3 & y_4 & y_5 & y_6& y_7 & y_8 & y_9\end{array} \right) \]
\end{small}

The encoder then checks $B_2$ and since ${\rm wt}(B_2)=34>32$, it replaces $B_2$ with $\phi_1(B_2)$ (as highlighted in blue), where ${\rm wt}(\phi(B_2))=30$, and sets $y_3=1$. 

\begin{small}
 \[A=\left( \begin{array}{ccccccccc}
0 & {\color{blue}{0}} & {\color{blue}{0}} & {\color{blue}{0}} & {\color{blue}{0}} & {\color{blue}{0}}& {\color{blue}{1}} &{\color{blue}{0}} & {\color{blue}{1}}\\
0 & {\color{blue}{1}} & {\color{blue}{1}} & {\color{blue}{0}} & {\color{blue}{0}} & {\color{blue}{0}}& {\color{blue}{0}} & {\color{blue}{0}} & {\color{blue}{1}}\\
0 & {\color{blue}{0}} & {\color{blue}{1}} & {\color{blue}{0}} & {\color{blue}{1}} & {\color{blue}{0}} & {\color{blue}{0}} & {\color{blue}{1}} & {\color{blue}{1}}\\
0 & {\color{blue}{0}} & {\color{blue}{0}} & {\color{blue}{1}} & {\color{blue}{1}} & {\color{blue}{0}} & {\color{blue}{0}} & {\color{blue}{0}} & {\color{blue}{0}}\\
0 & {\color{blue}{1}} & {\color{blue}{1}} & {\color{blue}{1}} & {\color{blue}{1}} & {\color{blue}{1}} & {\color{blue}{0}} & {\color{blue}{1}} & {\color{blue}{0}}\\
0 & {\color{blue}{0}} & {\color{blue}{1}} & {\color{blue}{0}} & {\color{blue}{0}} & {\color{blue}{1}} & {\color{blue}{1}} & {\color{blue}{0}} & {\color{blue}{1}}\\
0& {\color{blue}{1}} & {\color{blue}{1}} & {\color{blue}{0}} & {\color{blue}{0}} & {\color{blue}{0}} & {\color{blue}{0}} & {\color{blue}{0}} & {\color{blue}{1}}\\
0 & {\color{blue}{0}} & {\color{blue}{1}} & {\color{blue}{0}} & {\color{blue}{1}} & {\color{blue}{1}} & {\color{blue}{1}} & {\color{blue}{1}} & {\color{blue}{1}}\\
{\color{red}{1}} & y_2 & {\color{red}{1}} & y_4 & y_5 & y_6& y_7 & y_8 & y_9\end{array} \right) \]
\end{small}

The encoder then checks $C_1$ and $C_2$ as two subarrays of size $8\times 7$, and observe that their weights are both smaller than 28. It then sets $y_5=y_7=0$. Since ${\color{red}{y_1y_3y_5y_7=1100}}$, it sets ${\color{green}{y_2y_4y_6y_8=0011}}$ and sets $y_9=0$. The final output is as follows.

\begin{small}
 \[A=\left( \begin{array}{ccccccccc}
0 & 0 & 0 & 0 & 0 & 0& 1 & 0 & 1\\
0 & 1 & 1 & 0 & 0 & 0& 0 & 0 & 1\\
0 & 0 & 1 & 0 & 1 & 0& 0 & 1 & 1\\
0 & 0 & 0 & 1 & 1 & 0& 0 & 0 & 0\\
0 & 1 & 1 & 1 & 1 & 1& 0 & 1 & 0\\
0 & 0 & 1 & 0 & 0 & 1& 1 & 0 & 1\\
0 & 1 & 1 & 0 & 0 & 0& 0 & 0 & 1\\
0 & 0 & 1 & 0 & 1 & 1& 1 & 1 & 1\\
{\color{red}{1}} & {\color{green}{0}} & {\color{red}{1}} & {\color{green}{0}} &  {\color{red}{0}} & {\color{green}{1}}&  {\color{red}{0}} & {\color{green}{1}} & 0\end{array} \right) \]
\end{small}
\end{example}

\begin{remark} To obtain capacity-approaching codes, the construction can be further extended for $m=n-k$ when $k$ is no longer a constant. Since the redundancy of the encoder is $2(k+1)^2$, in oder to get the asymptotic rate of the encoder to be one, it is sufficient to require $k=o(n)$.
\end{remark}

\subsection{SRT for ${\rm Bal}_{\rm S}(n,m;\epsilon)$}

In this subsection, we show that SRT is an efficient method to encode ${\rm Bal}_{\rm S}(n,m;\epsilon)$. Similar to the results in Section IV, we show that for sufficiently large $n$, the coding method incurs at most one redundant bit. Recall the definition of $(\ell,\epsilon)$-r-forbidden and $(\ell,\epsilon)$-c-forbidden in Definition~\ref{forbidden-definition}.

\begin{lemma}\label{idea}
Suppose that $A$ is a binary array of size $n\times n$. If there is no $(m,\epsilon)$-r-forbidden in any row of $A$ or there is no $(m,\epsilon)$-c-forbidden in any column of $A$ then $A \in {\rm Bal}_{\rm S}(n,m;\epsilon)$.
\end{lemma}
\begin{proof}
Clearly, if there is no $(m,\epsilon)$-r-forbidden in any row of $A$, then every $m$ consecutive bits in each row forms an $\epsilon$-balanced sequence. If we consider any subarray of size $m\times m$, since each row is $\epsilon$-balanced, we then have the subarray is also $\epsilon$-balanced.
\end{proof}

According to Lemma~\ref{idea}, one may view an array of size $n\times n$ as a binary sequence of length $n^2$ and then uses our SRT coding method (as presented in Section IV) to enforce the $\epsilon$-balanced weight constraint over every $m$ consecutive bits. We summarise the result as follows. 

\begin{theorem}[Modified Theorem~\ref{theoremSRT}]\label{theoremSRT-modified}
Given $n>0,\epsilon\in (0,1/2)$. For $(1/\epsilon^2) \ln (n^2) \le m \le n$, there exists linear-time algorithms $\enc_{{\rm Bal}_{\rm S}(n,m;\epsilon)} :\{0,1\}^{n^2-1} \to\{0,1\}^{n^2}$ and $\dec_{{\rm Bal}_{\rm S}(n,m;\epsilon)}: {\rm Bal}_{\rm S}(n,m;\epsilon) \to \{0,1\}^{n^2-1}$ such that for all $\bx\in \{0,1\}^{n^2-1}$ if $A=\enc_{{\rm Bal}_{\rm S}(n,m;\epsilon)}(\bx)$, which is an array of size $n\times n$, then for every window $\bw$ of size $m$ of $A_i$, ${\rm wt}(\bw) \in [(1/2-\epsilon)m,(1/2+\epsilon)m]$ for all $1\le i\le n$. In other words, we have $A \in {\rm Bal}_{\rm S}(n,m;\epsilon)$. Furthermore, we have $\dec_{{\rm Bal}_{\rm S}(n,m;\epsilon)} \circ \enc_{{\rm Bal}_{\rm S}(n,m;\epsilon)}(\bx) \equiv \bx$ for all $\bx\in \{0,1\}^{n^2-1}$.
\end{theorem}

\begin{remark}
Given $\epsilon\in(0,1/2)$, Theorem~\ref{theoremSRT-modified} requires $(1/\epsilon^2) \ln (n^2) \le m \le n$, or $(2/\epsilon^2) \ln n \le m \le n$. Since $n\gg \ln n$, Theorem~\ref{theoremSRT-modified} works for a wide range of $m$ with respect to $n$.
\end{remark}

\section{Conclusion}
We have presented efficient encoding/decoding methods for two types of constraints over two-dimensional binary arrays: the $p$-bounded constraint and the $\epsilon$-balanced constraint. The constraint is enforced over either every row and every column, regarded as the 2D row/column (RC) constrained codes, or over every subarray, regarded as the 2D subarray constrained codes. The coding methods are based on: the divide and conquer algorithm and a modification of the Knuth's balancing technique, the sequence replacement technique, and the construction of antipodal matching as introduced in \cite{ROTH:ISIT}. For certain codes parameters, we have shown that there exist linear-time encoding/decoding algorithms that incur at most one redundant bit. 

To conclude, we discuss open problems and possible future directions of research.
\begin{enumerate}
\item {\em Study the channel capacity.} The capacity of the constraint channels are defined by 
\begin{align*}
{\bf c}_{\rm RC}(p) &\triangleq \lim_{n \to \infty} \frac{\log |{\rm B}_{\rm RC}(n;p)|}{n^2}, {\bf c}_{\rm RC}({\epsilon}) \triangleq \lim_{n \to \infty} 1/n^2 \log |{\rm Bal}_{\rm RC}(n;\epsilon)|, \text{ and } \\
{\bf c}_{\rm S}(m;p) &\triangleq \lim_{n \to \infty} \frac{\log |{\rm B}_{\rm S}(n,m;p)|}{n^2}, {\bf c}_{\rm S}(m;{\epsilon}) \triangleq \lim_{n \to \infty} 1/n^2 \log |{\rm Bal}_{\rm S}(n,m;\epsilon)|.
\end{align*}
In this work we show that ${\bf c}_{\rm RC}(p)=1$ for all $p\ge 1/2$, and ${\bf c}_{\rm RC}({\epsilon})=1$ for all $\epsilon$. On the other hand, the values ${\bf c}_{\rm S}(m;p)$, ${\bf c}_{\rm S}(m;\epsilon)$ remain unknown for fixed $m$, which is deferred to our future research work. Although we can design efficient encoders for ${\rm Bal}_{\rm S}(n,m;p)$ when $m =n-o(n)$ or ${\rm Bal}_{\rm S}(n,m;\epsilon)$ when $m \ge (2/\epsilon^2) \ln n$, a general construction for arbitrary values of $m$ remains as an open problem. 
\item {\em Combine the constrained encoders with error-correction capability.} To further reduce the error propagation during decoding procedure, we are interested in the problem of combining our proposed encoders with error-correction capability. Recently, the problem of correcting multiple criss-cross deletions (or insertions) in arrays has been investigated in \cite{2derror,2d-error}. A natural question is whether such codes can be modified and adapted for our encoders so that the output arrays are 2D constrained codes that are also capable of correcting deletions, insertions and substitutions.
\end{enumerate}


\begin{thebibliography}{99}

\bibitem{ISIT1} T. T. Nguyen, K. Cai, K. A. S. Immink and Y. M. Chee, ``Efficient Design of Capacity-Approaching Two-Dimensional Weight-Constrained Codes," {\em 2021 IEEE International Symposium on Information Theory (ISIT)}, 2021, pp. 2930-2935, doi: 10.1109/ISIT45174.2021.9517970.

\bibitem{ISIT2} T. T. Nguyen, K. Cai, H. M. Kiah, K. A. S. Immink, and Y. M. Chee, ``Using one redundant bit to construct two-dimensional almost-balanced codes", in {\em 2022 IEEE International Symposium on Information Theory (ISIT)}, 2022, to appear, accepted April 2022. 

\bibitem{roth:1999} R. Talyansky, T. Etzion, and R. M. Roth, ``Efficient Code Constructions for Certain Two-Dimensional Constraints", {\em IEEE Transactions on Information Theory}, vol. 45, no. 2, pp. 794-799, Mar. 1999.

\bibitem{D:1990} D. Psaltis, M. A. Neifeld, A. Yamamura, and S. Kobayashi, ``Optical memory disks in optical information processing," {\em Appl. Optics}, vol. 29, pp. 2038-2057, 1990.
\bibitem{J:1996} J. J. Ashley, M. Blaum, and B. H. Marcus, ``Report on coding techniques for holographic storage," {\em IBM Res. Rep. RJ 10013}, 1996.
\bibitem{D:1992} D. Brady and D. Psaltis, ``Control of volume holograms," {\em J. Opt. Soc. Am. A}, vol. 9, pp. 1167-1182, 1992.
\bibitem{Chen:2011} A. Chen, ``Accessibility of nano-crossbar arrays of resistive switching devices," {\em Proc. 11th IEEE Int. Conf. Nanotechnol.}, pp. 1767-1771, Aug. 2011.
\bibitem{T:2009} T. Raja and S. Mourad, ``Digital logic implementation in memristor-based crossbars," {\em Proc. Int. Conf. Commun. Circuits Syst. (ICCCAS)}, pp. 303-309, Jan. 2009.
\bibitem{S:2014} S. Kvatinsky, G. Satat, N. Wald, E. G. Friedman, A. Kolodny, and U. C. Weiser, ``Memristor-based material implication (imply) logic: Design principles and methodologies," {\em IEEE Trans. Very Large Scale Integr. (VLSI) Syst.}, vol. 22, no. 10, pp. 2054-2066, Oct. 2014.
\bibitem{R:2016} R. Ben Hur and S. Kvatinsky, ``Memory processing unit for in-memory processing," in {\em Proc. IEEE ACM Int. Symp. Nanosc. Archit.}, pp. 171-172, Jul. 2016.

\bibitem{ROTH:ISIT} E. Ordentlich and R. M. Roth, ``Low complexity two-dimensional weight-constrained codes," {\em 2011 IEEE International Symposium on Information Theory Proceedings}, St. Petersburg, 2011, pp. 149-153, doi: 10.1109/ISIT.2011.6033792.

\bibitem{O:2009} P. O. Vontobel, W. Robinett, P. J. Kuekes, D. R. Stewart, J. Straznicky, and R. S. Williams, ``Writing to and reading from a nano-scale crossbar memory based on memristors," {\em Nanotechnology}, vol. 20, no. 42, Sep. 2009, Art. no. 425204.

\bibitem{eitan:2016}  Y. Cassuto, S. Kvatinsky and E. Yaakobi, ``Information-Theoretic Sneak-Path Mitigation in Memristor Crossbar Arrays," in {\em IEEE Transactions on Information Theory}, vol. 62, no. 9, pp. 4801-4813, Sep. 2016, doi: 10.1109/TIT.2016.2594798.

\bibitem{zhong:2020} X. Zhong, K. Cai, G. Song and N. Raghavan, ``Deep Learning Based Detection for Mitigating Sneak Path Interference in Resistive Memory Arrays," 2020 IEEE International Conference on Consumer Electronics - Asia (ICCE-Asia), Seoul, 2020, pp. 1-4, doi: 10.1109/ICCE-Asia49877.2020.9277380.

\bibitem{G:2020} G. Song, K. Cai, X. Zhong, J. Yu, and J. Cheng, ``Performance Limit and Coding Schemes for Resistive Random-Access Memory Channels", {\em arXiv}, arXiv:2005.02601.

\bibitem{ROTH:2000} E. Ordentlich and R. M. Roth, ``Two-dimensional weight-constrained codes through enumeration bounds," in {\em IEEE Transactions on Information Theory}, vol. 46, no. 4, pp. 1292-1301, Jul. 2000, doi: 10.1109/18.850669.

\bibitem{vardy:1996} A. Vardy, M. Blaum, P. H. Siegel, and G. T. Sincerbox, ``Conservative arrays: Multidimensional modulation codes for holographic recording," {\em IEEE Trans. Inform. Theory}, vol. 42, pp. 227-230, 1996.

\bibitem{knuth:1986} D. E. Knuth, ``Efficient Balanced Codes", {\em IEEE Trans. Inform. Theory}, vol.IT-32, no. 1, pp. 51-53, Jan. 1986.

\bibitem{bose:1996} L. G. Tallini, R. M. Capocelli, and B. Bose,``Design of some new balanced codes," {\em IEEE Trans. Inform. Theory}, vol. 42, pp. 790-802, May 1996.


\bibitem{bound:2011} E. Ordentlich, F. Parvaresh and R. M. Roth, ``Asymptotic enumeration of binary matrices with bounded row and column weights," {\em 2011 IEEE International Symposium on Information Theory Proceedings}, St. Petersburg, 2011, pp. 154-158, doi: 10.1109/ISIT.2011.6033804.

\bibitem{immink:2010} K. A. S. Immink and J. H. Weber, ``Very Efficient Balanced Codes", {\em IEEE
J. Selected Areas Comms.}, vol. 28, no. 2, pp. 188-192, 2010.


\bibitem{alon:1988} N. Alon, E. E. Bergmann, D. Coppersmith, and A. M. Odlyzko, ``Balancing sets of vectors", {\em IEEE Trans. Inf. Theory}, vol. IT-34, no. 1, pp. 128-130, Jan. 1988.

\bibitem{TT:2020} T. Thanh Nguyen, K. Cai and K. A. Schouhamer Immink, ``Binary Subblock Energy-Constrained Codes: Knuth's Balancing and Sequence Replacement Techniques," {\em 2020 IEEE International Symposium on Information Theory (ISIT)}, Los Angeles, CA, USA, 2020, pp. 37-41, doi: 10.1109/ISIT44484.2020.9174430.

\bibitem{TT:DNA} T. T. Nguyen, K. Cai, K. A. Schouhamer Immink and H. M. Kiah, ``Capacity-Approaching Constrained Codes With Error Correction for DNA-Based Data Storage," in {\em IEEE Transactions on Information Theory}, vol. 67, no. 8, pp. 5602-5613, Aug. 2021, doi: 10.1109/TIT.2021.3066430.






















\bibitem{C:2020} C. D. Nguyen, V. K. Vu, and K. Cai, ``Two-Dimensional Weight-Constrained Codes for Crossbar Resistive Memory Arrays", {\em IEEE Commun. Lett.}, Early Access.

\bibitem{cover:1973} T. M. Cover, ``Enumerative source encoding," {\em IEEE Transactions on Information Theory}, vol. 19, no. 1, pp. 73-77, Jan. 1973. 

\bibitem{book} K. A. S. Immink, {\em Codes for Mass Data Storage Systems}, Second Edition, ISBN 90-74249-27-2, Shannon Foundation Publishers, Eindhoven, Nether- lands, 2004.
\bibitem{srt:2010} A. J. de Lind van Wijngaarden and K. A. S. Immink, ``Construction of Maximum Run-Length Limited Codes Using Sequence Replacement Techniques," {\em IEEE Journal on Selected Areas of Communications}, vol. 28, pp. 200-207, 2010.

\bibitem{srt:2019} O. Elishco, R. Gabrys, M. Medard, and E. Yaakobi, ``Repeated-Free Codes", {\em Proc. IEEE Int. Symp. Inf. Theory (ISIT)}, Paris, France, 2019.

\bibitem{2derror} R. Bitar, I. Smagloy, L. Welter, A. Wachter-Zeh and E. Yaakobi, ``Criss-Cross Deletion Correcting Codes," {\em 2020 International Symposium on Information Theory and Its Applications (ISITA)}, 2020, pp. 304-308.

\bibitem{2d-error} L. Welter, R. Bitar, A. Wachter-Zeh and E. Yaakobi, ``Multiple Criss-Cross Deletion-Correcting Codes," {\em 2021 IEEE International Symposium on Information Theory (ISIT)}, 2021, pp. 2798-2803, doi: 10.1109/ISIT45174.2021.9517743.


\end{thebibliography}
\end{document}